\documentclass{article}

\usepackage{amsmath}
\usepackage{amsthm}
\usepackage{amssymb}
\usepackage{amsfonts}
\usepackage{paralist}
\usepackage{fullpage}
\usepackage{url,hyperref}

\newtheorem{theorem}{Theorem}[section]
\newtheorem{lemma}[theorem]{Lemma}
\newtheorem{corollary}[theorem]{Corollary}

\newtheorem{propos}[theorem]{Proposition}

\newtheorem{definition}[theorem]{Definition}
\newtheorem{algorithm}[theorem]{Algorithm}

\newcommand{\E}{{\bf E}}
\newcommand{\Var}{{\bf Var}}
\newcommand{\RR}{\mathbb R}

\newcommand{\ZZ}{\mathbb Z}
\newcommand{\cI}{\mathcal I}
\newcommand{\cF}{\mathcal F}
\newcommand{\cD}{\mathcal D}

\def\b1{{\bf 1}}

\newcommand{\cond}{\ |\ }
\newcommand{\lcond}{\ \left|\ }
\newcommand{\zo}{\{0,1\}}
\newcommand{\on}{\{-1,1\}}
\newcommand{\zon}{\{0,1\}^n}
\newcommand{\sgn}{\mathsf{sign}}
\newcommand{\eps}{\epsilon}
\newcommand{\infl}{\mathsf{Infl}}
\newcommand{\U}{\mathcal{U}}
\newcommand{\equ}[1]{

\begin{equation}
#1
\end{equation}}

\newcommand{\alequn}[1]{\begin{align*} #1 \end{align*}}
\newcommand{\C}{{\mathcal C}}
\newcommand{\A}{{\mathcal A}}
\newcommand{\B}{{\mathcal B}}
\newcommand{\eat}[1]{}

\newcommand{\poly}{\mathrm{poly}}
\newcommand{\eg}{e.g.\ }
\newcommand{\fr}[1]{\frac{1}{#1}}
\newcommand{\I}{{\cal I}}

\newcommand{\F}{{\mathcal F}}
\newcommand{\etal}{{\em et al.\ }}

\newcommand{\ra}{\rangle}
\newcommand{\la}{\langle}

\title{Optimal Bounds on Approximation \\ of Submodular and XOS Functions by Juntas}

\author{Vitaly Feldman \\
IBM Research - Almaden \and Jan Vondr\'{a}k \\
IBM Research - Almaden\\
}

\begin{document}
\maketitle

\begin{abstract}
We investigate the approximability of several classes of real-valued functions by functions of a small number of variables ({\em juntas}). Our main results are tight bounds on the number of variables required to approximate a function $f:\zo^n \rightarrow [0,1]$ within $\ell_2$-error $\epsilon$ over the uniform distribution:
\begin{itemize}
\item If $f$ is submodular, then it is $\epsilon$-close to a function of $O(\frac{1}{\epsilon^2} \log \frac{1}{\epsilon})$ variables. This is an exponential improvement over previously known results \cite{FeldmanKV:13}. We note that $\Omega(\frac{1}{\epsilon^2})$ variables are necessary even for linear functions.
\item If $f$ is fractionally subadditive (XOS) it is $\epsilon$-close to a function of $2^{O(1/\epsilon^2)}$ variables. This result holds for all functions with low total $\ell_1$-influence and is a real-valued generalization of Friedgut's theorem for boolean functions. We show that $2^{\Omega(1/\epsilon)}$ variables are necessary even for XOS functions.
\end{itemize}

As applications of these results, we provide learning algorithms over the uniform distribution. For XOS functions, we give a PAC learning algorithm that runs in time $2^{1/\poly(\epsilon)} \poly(n)$. For submodular functions we give an algorithm in the more demanding PMAC learning model \cite{BalcanHarvey:12full} which requires a multiplicative $(1+\gamma)$ factor approximation with probability at least $1-\eps$ over the target distribution. Our uniform distribution algorithm runs in time $2^{1/\poly(\gamma\epsilon)} \poly(n)$. This is the first algorithm in the PMAC model that can achieve a constant approximation factor arbitrarily close to 1 for all submodular functions (even over the uniform distribution). It relies crucially on our bounds for approximation by juntas. As follows from the lower bounds in \cite{FeldmanKV:13} both of these algorithms are close to optimal. We also give applications for proper learning, testing and agnostic learning of these classes.

\end{abstract}

\thispagestyle{empty}
\newpage
\setcounter{page}{1}

\section{Introduction}
\label{sec:intro}

In this paper, we study the structure and learnability of several classes of real-valued functions over the uniform distribution on the Boolean hypercube $\zo^n$. The primary class of functions that we consider is the class of submodular functions. Submodularity, a discrete analog of convexity, has played an essential role in combinatorial optimization~\cite{E70,L83,Q95,F97,FFI00}.
Recently, interest in submodular functions has been revived by new applications
in algorithmic game theory as well as machine learning.
In machine learning, several applications \cite{GKS05,KGGK06,KSG08} have relied on the fact
that the information provided by a collection of sensors is a submodular function.
In algorithmic game theory, submodular functions have found application as {\em valuation functions}
with the property of diminishing returns \cite{LLN06,DS06,Vondrak08}. Along with submodular functions, other related classes have been studied in the algorithmic game theory context: coverage functions, gross substitutes, fractionally subadditive (XOS) functions, etc. It turns out that these classes are all contained in a broader class, that of {\em self-bounding functions}, introduced in the context of concentration of measure inequalities \cite{BoucheronLM:00}.  We refer the reader to Section~\ref{sec:prelims} for definitions and relationships of these classes.

Our focus in this paper is on {\em structural properties} of these classes of functions, specifically on their approximability by {\em juntas} (functions of a small number of variables) over the uniform distribution on $\zon$. Approximations of various function
classes by juntas is one of the fundamental topics in Boolean function analysis \cite{NisanSzegedy:92,Friedgut:98,Bourgain:02,FriedgutKN:02} with a growing number of applications in learning theory, computational complexity and algorithms \cite{DinurSafra:05, ChawlaKKRS:06, KrauthgamerRabani:06, ODonnellServedio:07, KhotRegev:08, GopalanMR:12, FeldmanKV:13}. A classical result in this area is Friedgut's theorem  \cite{Friedgut:98} which states that every boolean function $f$ is $\epsilon$-close to a function of $2^{O(\infl(f)/\epsilon^2)}$ variables, where $\infl(f)$ is the total influence of $f$ (see Sec.~\ref{sec:lowsens-prelims} for the formal definition). Such a result is not known for general real-valued functions, and in fact one natural generalization Freidgut's theorem is known not to hold \cite{ODonnellServedio:07}. However, it was recently shown \cite{FeldmanKV:13} that every submodular function with range $[0,1]$ is $\eps$ close in $\ell_2$-norm to a $2^{O(1/\epsilon^2)}$-junta. Stronger results are known in the special case when a submodular function only takes $k$ different values (for some small $k$). For this case Blais \etal prove existence of a junta of size $(k \log(1/\eps))^{O(k)}$ \cite{BlaisOSY:13manu} and Feldman \etal give a $(2^k/\eps)^5$ bound \cite{FeldmanKV:13}.

As in \cite{FeldmanKV:13}, our interest in approximation by juntas is motivated by applications to learning of submodular and XOS functions. The question of learning submodular functions from random examples was first formally considered by Balcan and Harvey \cite{BalcanHarvey:12full} who motivate it by learning of valuation functions. Reconstruction of submodular functions up to some multiplicative factor from value queries (which allow the learner to ask for the value of the function at any point) was also considered by Goemans \etal \cite{GHIM09}. These works and wide-spread applications of submodular functions have recently lead to significant attention to several additional variants of the problem of learning and testing submodular functions as well as their structural properties \cite{GuptaHRU:11,SV11,CheraghchiKKL:12,BadanidiyuruDFKNR:12,BalcanCIW:12,RaskhodnikovaYaroslavtsev:13,FeldmanKV:13,BlaisOSY:13manu}.
We survey related work in more detail in Sections \ref{sec:our-results} and \ref{sec:related-work}.

\subsection{Our Results}
\label{sec:our-results}
Our work addresses the following two questions: (i) what is the optimal size of junta that $\eps$-approximates a submodular function, and in particular whether the known bounds are optimal; (ii) which more general classes of real-valued functions can be approximated by juntas, and in particular whether XOS functions have such approximations.

In short, we provide the following answers: (i) For submodular functions with range $[0,1]$, the optimal $\eps$-approximating junta has size $\tilde{O}({1}/{\epsilon^2})$. This is an exponential improvement over the bounds in \cite{FeldmanKV:13,BlaisOSY:13manu}
which shows that submodular functions behave almost as linear functions (which are submodular) and are simpler than XOS functions which require a $2^{\Omega(1/\eps)}$-junta to approximate. This result is proved using new techniques.
(ii) All functions with range $[0,1]$ and constant total $\ell_1$-influence can be approximated in $\ell_2$-norm by a $2^{O(1/\eps^2)}$-junta. We show that this captures submodular functions, XOS and even self-bounding functions. This result is a real-valued generalization of Friedgut's theorem and is proved using the same technique.

We now describe these structural results formally and then describe new learning and testing algorithms that rely on them.

\subsubsection{Structural results}
Our main structural result is an approximation of submodular functions by juntas.
\begin{theorem}
\label{thm:submod-junta}
For any $\epsilon \in (0,\frac12)$ and any submodular function $f:\{0,1\}^n \rightarrow [0,1]$,
there exists a submodular function $g:\{0,1\}^n \rightarrow [0,1]$ depending only on a subset of variables $J \subseteq [n]$, $|J| = O(\frac{1}{\epsilon^2} \log \frac{1}{\epsilon})$, such that $\|f - g\|_2 \leq \epsilon$.
\end{theorem}

We also show that this result extends to arbitrary product distributions, with a dependence on the bias of the distribution (see Appendix~\ref{sec:product-distribution}). In the special case of submodular functions that take values in $\{0,1,\ldots,k\}$, our result can be simplified to give a junta of size $O(k\log(k/\eps))$  ($\eps$ being the disagreement probability). This is an exponential improvement over bounds in both \cite{FeldmanKV:13} and \cite{BlaisOSY:13manu} (see Corollary \ref{cor:submod-junta-pseudo} for a formal statement).

\medskip

\noindent{\bf Proof technique.}  Our proof is based on a new procedure that selects variables to be included in the approximating junta for a submodular function $f$. We view the hypercube $\{0,1\}^n$ as subsets of $\{1,2,\ldots,n\}$ and refer to $f(S \cup \{i\}) - f(S)$ as the marginal value of variable $i$ on set $S$. Iteratively, we add a variable $i$ if its marginal value is large enough with probability at least $1/2$ taken over sparse random subsets of the variables that are already chosen. 
One of the key pieces of the proof is the use of a ``{\em boosting lemma}\footnote{The terminology comes from \cite{GoemansVondrak:06} and has no connection with the notion of boosting in machine learning.}" on down-monotone events of  Goemans and Vondr\'ak \cite{GoemansVondrak:06}. We use it to show that our criterion for selection of the variables implies that with very high probability over a random and uniform choice of a subset of the selected variables, the marginal value of each of the variables that are excluded is small. The probability of having small marginal value is high enough to apply a union bound over all excluded variables. Bounded marginal values are equivalent to the function being Lipschitz in all the excluded variables which allows us to apply {\em concentration of Lipschitz submodular functions} to replace the functions of excluded variables by constants. Concentration bounds for submodular functions were first given by Boucheron \etal \cite{BoucheronLM:00} and are also a crucial component of some of the prior works in this area \cite{BalcanHarvey:12full,GuptaHRU:11,FeldmanKV:13}.

One application of this procedure allows us to reduce the number of variables from $n$ to $O(\frac{1}{\epsilon^2} \log \frac{n}{\epsilon})$. This process can be repeated until the number of variables becomes $O(\frac{1}{\epsilon^2} \log \frac{1}{\epsilon})$.

Using a more involved argument based on the same ideas we show that monotone submodular functions can with high probability be {\em multiplicatively} approximated by a junta. Formally, $g$ is a multiplicative $(\alpha,\epsilon)$-approximation to $f$ over a distribution $D$, if $\Pr_D[f(x) \leq g(x) \leq \alpha f(x)] \geq 1-\epsilon$. In the PMAC learning model, introduced by Balcan and Harvey \cite{BalcanHarvey:12full} a learner has to output a hypothesis that multiplicatively $(\alpha,\epsilon)$-approximates the unknown function. It is a relaxation of the worst case multiplicative approximation used in optimization but is more demanding than the $\ell_1$/$\ell_2$-approximation that is the main focus of our work. We prove the following:

\begin{theorem}
\label{thm:PMAC-junta}
For every monotone submodular function $f:\{0,1\}^n \rightarrow \RR_+$ and every $\gamma, \epsilon \in (0,1)$, there is a monotone submodular function $h:\{0,1\}^J \rightarrow \RR_+$ depending only on a subset of variables $J \subseteq [n], |J| = O(\frac{1}{\gamma^2} \log \frac{1}{\gamma \epsilon} \log \frac{1}{\epsilon})$ such that $h$ is a multiplicative $(1+\gamma,\epsilon)$-approximation of $f$ over the uniform distribution.
\end{theorem}

We then show that broader classes of functions such as XOS and self-bounding can also be approximated by juntas, although of an exponentially larger size. We denote by $\infl^1(f)$ the total $\ell_1$-influence of $f$ and by $\infl^2(f)$ the total $\ell_2^2$-influence of $f$ (see Sec.~\ref{sec:lowsens-prelims} for definitions). We prove the result via the following generalization of the well-known Friedgut's theorem for boolean functions.
\begin{theorem}
\label{thm:Friedgut-junta-intro}
Let $f:\zo^n \rightarrow \RR$ be any function and $\epsilon>0$. There exists a function $g:\{0,1\}^n \rightarrow \RR$ depending only on a subset of variables $J \subseteq [n]$, $|J| = 2^{O(\infl^2(f)/\epsilon^2)} \cdot (\infl^1(f))^3/\eps^4$ such that $\|f - g\|_2 \leq \epsilon$.
For a submodular, XOS or self-bounding $f:\zo^n \rightarrow [0,1]$, $\infl^2(f) \leq \infl^1(f) = O(1)$, giving $|J| = 2^{O(1/\epsilon^2)}$.
\end{theorem}
Friedgut's theorem gives approximation by a junta of size $2^{O(\infl(f)/\epsilon^2)}$ for a boolean $f$. For a boolean function, the total influence $\infl(f)$ (also referred to as average sensitivity) is equal to both $\infl^1(f)$ and $\infl^2(f)$ (up to a fixed constant factor). Previously it was observed that Friedgut's theorem is not true if $\infl^2(f)$ is used in place of $\infl(f)$ in the statement \cite{ODonnellServedio:07}. However we show that with an additional factor which is just polynomial in $\infl^1(f)$ one can obtain a generalization.
O'Donnell and Servedio \cite{ODonnellServedio:07} generalized the Friedgut's theorem to bounded discretized real-valued functions. They prove a bound of $2^{O(\infl^2(f)/\epsilon^2)} \cdot \gamma^{-O(1)}$, where $\gamma$ is the discretization step. This special case is easily implied by our bound. Technically, our proof is a simple refinement of the proof of Friedgut's theorem.

The second component of this result is a simple proof that self-bounding functions (and hence submodular and XOS) have constant total $\ell_1$-influence. An immediate implication of this fact alone is that self-bounding functions can be approximated by functions of Fourier degree $O(1/\eps^2)$. For the special case of submodular functions this was proved by Cheraghchi \etal also using Fourier analysis, namely, by bounding the noise stability of submodular functions \cite{CheraghchiKKL:12}. Our more general proof is substantially simpler.

We show that this result is almost tight, in the sense that even for XOS functions 
 $2^{\Omega(1/\epsilon)}$ variables are necessary for an $\epsilon$-approximation in $\ell_1$ (see Thm.~\ref{thm:XOS-example}).
Thus we obtain an almost complete picture, in terms of how many variables are needed to achieve an $\epsilon$-approximation depending on the target function --- see Figure~\ref{fig:table}.


\begin{figure}
\centering
$\begin{array}{|| c || c | c ||} \hline
\mbox{Class of functions} & \mbox{junta size lower bound} & \mbox{junta size upper bound} \\
\hline \hline
\mbox{linear} & \Omega(1/\epsilon^2) \mbox{\ [Folkl., see Lem.~\ref{lem:linear-example}]} & O(1/\epsilon^2) \mbox{\ [Folkl.]} \\
\hline
\mbox{coverage} & \mbox{as above}   & O(1/\epsilon^2)\ \cite{FeldmanK14} \\
\hline
\mbox{submodular} & \mbox{as above} & O(1/\epsilon^2 \cdot \log (1/\epsilon)) \mbox{\ [Thm.~\ref{thm:submod-junta}]} \\
\hline
\mbox{XOS and self-bounding} & 2^{\Omega(1/\epsilon)} \mbox{\ [Thm.~\ref{thm:XOS-example}]} & 2^{O(1/\epsilon^2)} \mbox{\ [Thm.~\ref{thm:Friedgut-junta-intro}]}\\
\hline
\mbox{constant total $\ell_1$-influence} & 2^{\Omega(1/\epsilon)}\ \cite{Friedgut:98} & 2^{O(1/\epsilon^2)} \mbox{\ [Thm.~\ref{thm:Friedgut-junta-intro}]}\\
\hline
\mbox{constant total $\ell_2^2$-influence} & \Omega(n)\ \cite{ODonnellServedio:07} & n \\
\hline
\end{array} $
\caption{\small Overview of junta approximations: bounds on the size of a junta achieving an $\epsilon$-approximation in $\ell_2$ for a function with range $[0,1]$. \label{fig:table}}
\end{figure}

\subsubsection{Applications}
We provide several applications of our structural results to learning and testing.  These applications are based on new algorithms as well as standard approaches to learning over the uniform distribution.

For submodular functions our main application is a PMAC learning algorithm over the uniform distribution.

\begin{theorem}
\label{thm:pmac-learn-submod-intro}
There exists an algorithm $\A$ that given $\gamma,\eps \in (0,1]$ and access to random and uniform examples of a submodular function $f:\zo^n \rightarrow \RR_+$, with probability at least $2/3$, outputs a function $h$ which is a multiplicative $(1+\gamma,\eps)$-approximation to $f$ (over the uniform distribution). Further, $\A$ runs in time $\tilde{O}(n^2) \cdot 2^{\tilde{O}(1/(\eps\gamma)^2)}$ and uses $\log(n) \cdot 2^{\tilde{O}(1/(\eps\gamma)^2)}$ examples.
\end{theorem}

We remark that this algorithm works even for non-monotone submodular functions and does not in fact rely on our multiplicative-approximation junta result (Theorem~\ref{thm:PMAC-junta}, which works only for monotone submodular functions). Instead, we boostrap the $\ell_2$-approximation result (Theorem~\ref{thm:submod-junta}) as follows.
Theorem~\ref{thm:submod-junta} guarantees an $\ell_2$-approximating junta of size $\tilde{O}(1/\eps^2)$. The main challenge here is that the criterion for including variables used in the proof of Theorem \ref{thm:submod-junta} cannot be (efficiently) evaluated using random examples alone. Instead we give a general algorithm to find a larger approximating junta whenever an approximating junta exists. This algorithm relies only on submodularity of the function and in our case finds a junta of size $\tilde{O}(1/\eps^5)$. From there one can easily use brute force to find a $\tilde{O}(1/\eps^2)$-junta in time $2^{\tilde{O}(1/\eps^2)}$. 

We show that using the function $g$ returned by this building block we can partition the domain into $2^{\tilde{O}(1/\epsilon^2)}$ subcubes  such that on a constant fraction of those subcubes $g$ gives a multiplicative $(1+\gamma,\eps)$ approximation. We then apply the building block recursively for $O(\log(1/\eps))$ levels.

In addition, the algorithm for finding close-to-optimal $\ell_2$-approximating junta allows us to learn properly (by outputting a submodular function) in time $2^{\tilde{O}(1/\eps^2)} \poly(n)$. Using a standard transformation we can also test whether the input function is submodular or $\epsilon$-far (in $\ell_1$) from submodular, in time $2^{\tilde{O}(1/\epsilon^2)} \cdot \poly(n)$ and using just $2^{\tilde{O}(1/\epsilon^2)} + \poly(1/\eps) \log n$ random examples. (Using earlier results, this would have been possible only in time doubly-exponential in $\epsilon$.) We give the details of these results in Section \ref{sec:applications-learning}.

For XOS functions, we give a PAC learning algorithm with $\ell_2$ error using the junta and low Fourier degree approximation for self-bounding functions (Theorem~\ref{thm:Friedgut-junta-intro}).

\begin{theorem}
\label{thm:XOS-learning-intro}
There exists an algorithm $\A$ that given $\eps > 0$ and access to random uniform examples of an XOS function $f:\zo^n \rightarrow [0,1]$, with probability at least $2/3$, outputs a function $h$, such that $\|f-h\|_2 \leq \epsilon$. Further, $\A$ runs in time $2^{O(1/\eps^4)} \poly(n)$ and uses $2^{O(1/\eps^4)} \log n$ random examples.
\end{theorem}
In this case the algorithm is fairly standard: we use the fact that XOS functions are monotone and hence their influential variables can be detected from random examples (as for example in \cite{Servedio:04mondnf}). Given the influential variables we can exploit the low Fourier degree approximation to find a hypothesis using $\ell_2$ regression over the low degree parities (as done in \cite{FeldmanKV:13}).

This algorithm naturally extends to any {\em monotone} real-valued function of low total $\ell_1$-influence, of which XOS functions are a special case. Using the algorithm in Theorem \ref{thm:XOS-learning-intro} we also obtain a PMAC-learning algorithm for XOS functions using the same approach as we used for submodular functions. However the dependence of the running time and sample complexity on $1/\gamma$ and $1/\epsilon$ is doubly-exponential in this case (see Cor.~\ref{cor:pmac-learn-xos} for details). To our knowledge, this is the first PMAC learning algorithm for XOS functions that can achieve constant approximation factor in polynomial time for all XOS functions.

\smallskip
\noindent {\bf Organization.} We present a detailed discussion of the classes of functions that we consider and technical preliminaries in Section~\ref{sec:prelims}. The proof of our main structural result (Thm.~\ref{thm:submod-junta}) is presented in Section \ref{sec:submod-junta}. Its extension to multiplicative approximation of monotone submodular functions (Thm.~\ref{thm:PMAC-junta}) is given in Section \ref{sec:PMAC-junta}. An extension to the case of general product distributions is presented in Appendix~\ref{sec:product-distribution}. In Section \ref{sec:low-sensitivity} we give the proof of real-valued generalization of Friedgut's theorem (Thm.~\ref{thm:Friedgut-junta-intro}). Section \ref{sec:lower-bounds} gives examples of functions that prove tightness of our bounds for submodular and XOS functions. The details of our algorithmic applications to PAC and PMAC learning are in Section \ref{sec:applications-learning}.
We state several implications of our structural results to agnostic learning and testing in Section \ref{sec:additinal-app}.

\subsection{Related Work}
\label{sec:related-work}
Reconstruction of submodular functions up to some multiplicative factor (on every point) from value queries was first considered by Goemans \etal \cite{GHIM09}. They show a polynomial-time algorithm for reconstructing monotone submodular functions with $\tilde{O}(\sqrt{n})$-factor approximation and prove a nearly matching lower-bound. This was extended to the class of all subadditive functions in \cite{BadanidiyuruDFKNR:12} which studies
small-size approximate representations of valuation functions (referred to as {\em sketches}). Theorem \ref{thm:PMAC-junta} shows that allowing an $\eps$ error probability (over the uniform distribution) makes it possible to get a multiplicative $(1+\gamma)$-approximation using a $\poly(1/\gamma,\log{(1/\eps))}$-sized sketch. This sketch can be found in polynomial time using value queries (see Section \ref{sec:PMAC-junta}).

Balcan and Harvey initiated the study of learning submodular functions from random examples coming from an unknown distribution and introduced the PMAC learning model described above \cite{BalcanHarvey:12full}. They give an O($\sqrt{n}$)-factor PMAC learning algorithm and show an information-theoretic $\Omega(\sqrt[3]{n})$-factor impossibility result for submodular functions. Subsequently, Balcan \etal gave a distribution-independent PMAC learning algorithm for XOS functions that achieves an $\tilde{O}(\sqrt{n})$-approximation and showed that this is essentially optimal \cite{BalcanCIW:12}. They also give a PMAC learning algorithm in which the number of clauses defining the target XOS function determines the running time and the approximation factor that can be achieved (for polynomial-size XOS functions it implies $O(n^\beta)$-approximation factor in time $n^{O(1/\beta)}$ for any $\beta > 0$).

The lower bound in \cite{BalcanHarvey:12full} also implies hardness of learning of submodular function with $\ell_1$(or $\ell_2$)-error: it is impossible to learn a submodular function $f:\{0,1\}^n \rightarrow [0,1]$ in $\poly(n)$ time within any nontrivial $\ell_1$-error over general distributions. We emphasize that these strong lower bounds rely on a very specific distribution concentrated on a sparse set of points, and show that this setting is very different from the setting of uniform/product distributions which is the focus of this paper.

For product distributions, Balcan and Harvey show that 1-Lipschitz monotone submodular functions of minimum nonzero value at least $1$ have concentration properties implying a PMAC algorithm with a multiplicative $(O(\log \frac{1}{\eps}),\eps)$-approximation \cite{BalcanHarvey:12full}. The approximation is by a constant function and the algorithm they give approximates the function by its mean on a small sample. Since a constant is a function of $0$ variables, their result can be viewed as an extreme case of approximation by a junta. 
Our result gives multiplicative $(1+\gamma,\epsilon)$-approximation for arbitrarily small $\gamma,\epsilon>0$. The main point of Theorem \ref{thm:PMAC-junta}, perhaps surprising, is that the number of required variables grows only polynomially in $1/\gamma$ and logarithmically in $1/\epsilon$.

Learning of submodular functions with additive rather than multiplicative guarantees over the uniform distribution was first considered by Gupta \etal who were motivated by applications in private data release \cite{GuptaHRU:11}. They show that submodular functions can be $\epsilon$-approximated by a collection of $n^{O(1/\epsilon^2)}$ $\epsilon^2$-Lipschitz submodular functions. Concentration properties imply that each $\epsilon^2$-Lipschitz submodular function can be $\epsilon$-approximated by a constant. This leads to a learning algorithm running in time $n^{O(1/\epsilon^2)}$, which however requires value queries in order to build the collection. Cheraghchi \etal use an argument based on noise stability to show that submodular functions
can be approximated in $\ell_2$ by functions of Fourier degree $O(1/\eps^2)$ \cite{CheraghchiKKL:12}. This leads to an $n^{O(1/\eps^2)}$ learning algorithm which uses only random examples and, in addition, works in the agnostic setting. Most recently, Feldman \etal show that the decomposition from \cite{GuptaHRU:11} can be computed by a low-rank binary decision tree \cite{FeldmanKV:13}. They then show that this decision tree can then be pruned to obtain depth $O(1/\eps^2)$ decision tree that approximates a submodular function. This construction implies approximation by a $2^{O(1/\eps^2)}$-junta of Fourier degree $O(1/\eps^2)$. They used these structural results to give a PAC learning algorithm running in time $\poly(n) \cdot 2^{O(1/\eps^4)}$. Note that our multiplicative $(1+\gamma,\eps)$-approximation in this case implies $O(\gamma + \eps)$ $\ell_2$-error (but $\ell_2$-error gives no multiplicative guarantees). In \cite{FeldmanKV:13} it is also shown that $2^{\Omega(\eps^{-2/3})}$ random examples (or even value queries) are necessary to PAC learn monotone submodular functions to $\ell_1$-error of $\eps$. This implies that our learning algorithms for submodular and XOS functions cannot be substantially improved.

In a recent work, Raskhodnikova and Yaroslavtsev consider learning and testing of submodular functions taking values in the range $\{0,1,\ldots,k\}$ (referred to as {\em pseudo-Boolean}) \cite{RaskhodnikovaYaroslavtsev:13}. The error of a hypothesis in their framework is the probability that the hypothesis disagrees with the unknown function. They build on the approach from \cite{GuptaHRU:11} to show that pseudo-Boolean submodular functions can be expressed as $2k$-DNF and then apply Mansour's algorithm for learning DNF \cite{Mansour:95} to obtain a $\poly(n) \cdot k^{O(k \log{k/\eps})}$-time PAC learning algorithm using value queries. In this special case the results in \cite{FeldmanKV:13} give approximation of submodular functions by junta of size $\poly(2^k/\eps)$ and $\poly(2^k/\eps,n)$ PAC learning algorithm from random examples. In an independent work, Blais \etal prove existence of a junta of size $(k \log(1/\eps))^{O(k)}$ and use it to give an algorithm for testing submodularity using $(k \log(1/\eps))^{\tilde{O}(k)}$ value queries \cite{BlaisOSY:13manu}.

It is interesting to remark that several largely unrelated methods point to approximating junta being of exponential size, namely, pruned decision trees in \cite{FeldmanKV:13}; Friedgut's theorem based analysis in this work; two Sunflower lemma-style arguments in \cite{BlaisOSY:13manu}. However, unexpectedly (at least for the authors), a polynomial-size junta suffices.

Previously, approximations by juntas of size polynomial in $1/\eps$ were only known in some simple special cases of submodular functions. Boolean submodular functions are disjunctions and hence, over the uniform distribution, can be approximated by an $O(\log(1/\eps))$-junta. It can be easily seen that linear functions are approximable by $O(1/\eps^2)$-juntas. Coverage functions which are non-negative linear combinations of monotone disjunctions have been recently shown to be approximable by $O(1/\eps^2)$-juntas \cite{FeldmanK14}. More generally, for Boolean functions the results in \cite{DiakonikolasServedio:09} imply that linear threshold functions with constant total influence can be $\eps$-approximated by a junta of size polynomial in $1/\eps$. In both \cite{DiakonikolasServedio:09} and \cite{FeldmanK14} the techniques are unrelated to ours. 

\section{Preliminaries} \label{sec:prelims}

\subsection{Classes of valuation functions}

Let us describe several classes of functions on the discrete cube, which can be also equivalently viewed as set functions. The functions in these classes share some form of the property of ``forbidden complementarities" --- e.g., $f(\{a,b\})$ cannot be more than $f(\{a\}) + f(\{b\})$. These functions could be monotone or non-monotone; we call a function monotone if $f(S) \leq f(T)$ whenever $S \subset T$.



\noindent{\bf Linear functions.}
Linear (or additive) functions are functions in the form $f(S) = \sum_{i \in S} a_i$. This is the smallest class in the hierarchy that we consider here.

\noindent{\bf Submodular functions.}
Submodular functions are defined by the condition $f(A \cup B) + f(A \cap B) \leq f(A) + f(B)$ for all $A,B$.
A monotone submodular function can be viewed as a valuation on sets with the property of {\em diminishing returns}:
the marginal value of an element, $f_S(i) = f(S \cup \{i\}) - f(S)$, cannot increase if we enlarge the set $S$.
Non-monotone submodular functions play a role in combinatorial optimization, primarily as generalizations
of the {\em cut function} in a graph, $c(S) = |E(S, \bar{S})|$, which is known to be submodular. Another important subclass of monotone submodular functions is the class of {\em rank functions of matroids}: $r(S) = \max \{|I|: I \in \cI, I \subseteq S \}$, where $\cI$ is the family of independent sets in a matroid. In fact, it is known that a function of this type is submodular if and only if $\cI$ forms a matroid.

\noindent{\bf Fractionally subadditive functions (XOS).}
A set function $f$ is fractionally subadditive if $f(A) \leq \sum \beta_i f(B_i)$ whenever $\beta_i \geq 0$
 and $\sum_{i:a \in B_i} \beta_i \geq 1 \ \forall a \in A$.

This class is broader than that of (nonnegative) monotone submodular functions (but does not contain non-monotone functions, since fractionally-subadditive functions are monotone by definition). For fractionally subadditive functions such that $f(\emptyset) = 0$,
there is an equivalent definition known as ``XOS" or maximum of non-negative linear functions \cite{Feige:06}:
$f$ is XOS iff $f(S) = \max_{i \in [m]} \sum_{j \in S} w_{ij}$, where $m$ any positive integer and $w_{ij}$'s are arbitrary non-negative real-valued weights (note that for every $i$, $g_i(S) = \sum_{j \in S} w_{ij}$ is a non-negative linear function).

It is instructive to consider again the example of rank functions: $r(S) = \max \{|I|: I \in \cI, I \subseteq S \}$. As we mentioned,
$r(S)$ is submodular exactly when $\cI$ forms a matroid. In contrast, $r(S)$ is XOS for {\em any down-closed} set system $\cI$ (satisfying $A \subset B \in \cI \Rightarrow A \in \cI$;
this follows from an equivalent formulation of a rank function for down-closed set systems, $r(S) = \max \{ |S \cap I|: I \in \cI\}$). In this sense, XOS is a significantly broader class than submodular functions. Another manifestation of this fact is that optimization problems like $\max \{f(S): |S| \leq k\}$ admit constant-factor approximation algorithms using polynomially many {\em value queries} to $f$ when $f$ is submodular, but no such algorithms exist for XOS functions.

\noindent{\bf Subadditive functions.}
Subadditive functions are defined by the condition $f(A \cup B) \leq f(A) + f(B)$ for all $A,B$.
Subadditive functions are more general than submodular and fractionally subadditive functions.
In fact, subadditive functions are in some sense much less structured than fractionally subadditive  functions.
It is easy to verify that every function $f:2^N \rightarrow \{1,2\}$ is subadditive.
While submodular and fractionally subadditive functions satisfy ``dimension-free" concentration bounds, this is not true for subadditive functions (see \cite{Vondrak10} for more details).

\noindent{\bf Self-bounding functions.}
Self-bounding functions were defined by Boucheron, Lugosi and Massart
\cite{BoucheronLM:00} and further generalized by McDiarmid and Reed \cite{MR06} as a unifying class of functions that enjoy strong concentration properties.
Self-bounding functions are defined generally on product spaces $X^n$; here we restrict our attention to the hypercube,
i.e.~the case where $X = \{0,1\}$.
We identify functions on $\zo^n$ with set functions on $N = [n]$ in a natural way.
By $\bf 0$ and $\b1$, we denote the all-zeroes and all-ones vectors in $\{0,1\}^n$ respectively (corresponding to $\emptyset$ and $N$ sets).

\begin{definition}
For a function $f:\zo^n \rightarrow \RR$ and any $x \in \zo^n$,  let $\min_{x_i}f(x) = \min{\{f(x), f(x\oplus e_i)\}}$.
Then $f$ is $(a,b)$-self-bounding, if
for all $x \in \zo^n$ and $i \in [n]$,
\begin{eqnarray}
 f(x) - \min_{x_i} f(x) & \leq & 1, \label{eq:Lipschitz} \\
 \sum_{i=1}^{n} (f(x) - \min_{x_i} f(x)) & \leq & a f(x) + b.
\end{eqnarray}
\end{definition}


In this paper, we are primarily concerned with $(a,0)$-self-bounding functions, to which we also refer as $a$-self-bounding functions. Note that the definition implies that $f(x) \geq 0$ for every $a$-self-bounding function. Self-bounding functions include ($1$-Lipschitz) fractionally subadditive functions. To subsume $1$-Lipschitz non-monotone submodular functions, it is sufficient to consider the slightly more general $2$-self-bounding functions --- see \cite{Vondrak10}. The $1$-Lipschitz condition will not play a role in this paper, as we normalize functions to have values in the $[0,1]$ range.

Self-bounding functions satisfy {\em dimension-free concentration bounds}, based on the entropy method of Boucheron, Lugosi and Massart \cite{BoucheronLM:00}. Currently this is the most general class of functions known to satisfy such concentration bounds. The entropy method for self-bounding functions is general enough to rederive bounds such as Talagrand's concentration inequality.
An example of a self-bounding function (related to applications of Talagrand's inequality) is a function with the property of {\em small certificates}: $f:X^n \rightarrow \ZZ_+$ has small certificates, if it is 1-Lipschitz and whenever $f(x) \geq k$, there is a set of coordinates $S \subseteq [n]$, $|S|=k$, such that if $y|_S = x|_S$, then $f(y) \geq k$. Such functions often arise in combinatorics, by defining $f(x)$ to equal the maximum size of a certain structure appearing in $x$.
Another well-studied class of self-bounding functions arises from Rademacher averages which are widely used to measure the complexity of model classes in statistical learning theory \cite{Koltchinskii01,BartlettMendelson:02}. See \cite{BoucheronLB03} for a more detailed discussion and additional examples.

The definition of self-bounding functions is more symmetric than that of submodular functions: note that the definition does not change if we swap the meaning of $0$ and $1$ for any coordinate. This is a natural property in the setting of machine learning; the learnability of functions on $\zo^n$ should not depend on switching the meaning of $0$ and $1$ for any particular coordinate.

\subsection{Norms and discrete derivatives}

The $\ell_1$ and $\ell_2$-norms of $f:\zon\rightarrow \RR$ are defined by $\|f\|_1 =  \E_{x \sim \U} [|f(x)|]$ and $\|f\|_2 =  (\E_{x \sim \U} [f(x)^2])^{1/2}$, respectively, where $\U$ is the uniform distribution.

\begin{definition}[Discrete derivatives]
For $x \in \zon$, $b \in \zo$ and $i \in n$, let $x_{i\leftarrow b}$ denote the vector in $\zo^n$ that equals $x$  with $i$-th coordinate set to $b$. For a function $f:\zon \rightarrow \RR$ and index $i \in [n]$ we define
$\partial_i f(x) = f(x_{i\leftarrow 1}) - f(x_{i\leftarrow 0})$.
We also define $\partial_{i,j} f(x) = \partial_i \partial_j f(x)$.
\end{definition}

A function is monotone (non-decreasing) if and only if for all $i\in [n]$ and $x\in \zo^n$, $\partial_i f(x) \geq 0$.
For a submodular function, $\partial_{i,j} f(x) \leq 0$, by considering the submodularity condition for $x_{i \leftarrow 0, j \leftarrow 0}$, $x_{i \leftarrow 0, j \leftarrow 1}$, $x_{i \leftarrow 1, j \leftarrow 0}$, and  $x_{i \leftarrow 1, j \leftarrow 1}$.

\smallskip
\noindent{\bf Absolute error vs.~error relative to norm:}
In our results, we typically assume that the values of $f(x)$ are in a bounded interval $[0,1]$,
and our goal is to learn $f$ with an additive error of $\epsilon$. Some prior work considered an error relative to the norm of $f$, for example at most $\epsilon \|f\|_1$ \cite{CheraghchiKKL:12}. In fact, it is known that for a non-negative submodular or XOS function $f$, $\|f\|_1 = \E[f] \geq \frac14 \|f\|_\infty$ \cite{Feige:06,FMV07} and hence this does not make much difference. If we scale $f(x)$ by $\frac{1}{4 \|f\|_1}$, we obtain a function with values in $[0,1]$ and learning the original function within an additive error of $\epsilon \|f\|_1$ is equivalent to learning the scaled function within an error of $\epsilon/4$.


\section{Junta Approximations of Submodular Functions}

First we turn to the class of submodular functions and their approximations by functions of a small number of variables.

\subsection{Additive Approximation For Submodular Functions}
\label{sec:submod-junta}

Here we prove Theorem~\ref{thm:submod-junta}, a bound of $\tilde{O}(1/\epsilon^2)$ on the size of a junta needed to approximate a submodular function bounded by $[0,1]$ within an additive error of $\epsilon$.
The core of our proof is the following (seemingly weaker) statement. We remark that in this paper all logarithms are base 2.

\begin{lemma}
\label{lem:submod-reduce}
For any $\epsilon \in (0,\frac12)$ and any submodular function $f:\{0,1\}^J \rightarrow [0,1]$,
there exists a submodular function $h:\{0,1\}^J \rightarrow [0,1]$ depending only on a subset of variables $J' \subseteq J$, $|J'| \leq \frac{128}{\epsilon^2} \log \frac{16|J|}{\epsilon^2}$, such that $\|f - h\|_2 \leq \frac12 \epsilon$.
\end{lemma}

Note that if $|J|=n$ and $\epsilon=\Omega(1)$, Lemma~\ref{lem:submod-reduce} reduces the number of variables to $O(\log n)$ rather than a constant. However, we show that this is enough to prove Theorem~\ref{thm:submod-junta}, effectively by repeating this argument. In fact, it was previously shown \cite{FeldmanKV:13} that submodular functions can be $\epsilon$-approximated by functions of $2^{O(1/\epsilon^2)}$ variables. One application of Lemma~\ref{lem:submod-reduce} to this result brings the number of variables down to $\tilde{O}(\frac{1}{\epsilon^4})$, and another repetition of the same argument brings it down to $O(\frac{1}{\epsilon^2} \log \frac{1}{\epsilon})$. This is a possible way to prove Theorem~\ref{thm:submod-junta}. Nevertheless, we do not need to rely on this previous result, and we can derive Theorem~\ref{thm:submod-junta} directly from Lemma~\ref{lem:submod-reduce} as follows.

\begin{proof}[Proof of Theorem~\ref{thm:submod-junta}]
Let $f:\{0,1\}^n \rightarrow [0,1]$ be a submodular function. We shall prove a bound of $|J| \leq \frac{4000}{\epsilon^2} \log \frac{1}{\epsilon}$ for the size of the approximating junta.

Observe that this bound holds trivially for $\epsilon \leq n^{-1/2}$, because then we are allowed to choose $J = [n]$. For contradiction, suppose that there is $\epsilon \in (n^{-1/2},1/2)$ for which the statement of Theorem~\ref{thm:submod-junta} does not hold. Let ${\cal E} \subseteq (n^{-1/2}, 1/2)$ be the set of all $\epsilon$ for which the statement does not hold, and pick an $\epsilon \in {\cal E}$ such that $\epsilon < 2 \inf {\cal E}$. Then, the statement still holds for $\epsilon_2 = \epsilon^2 < \frac12 \epsilon$.

By the statement of Theorem~\ref{thm:submod-junta} for $\epsilon_2$, there is a subset of variables $J$ of size $|J| \leq \frac{4000}{\epsilon_2^2} \log \frac{1}{\epsilon_2} = \frac{4000}{\epsilon^4} \log \frac{1}{\epsilon^2} \leq \frac{2^{13}}{\epsilon^5}$ and a submodular function $g$ depending only on $J$, such that $\|f - g\|_2 \leq \epsilon_2 \leq \frac12 \epsilon$. Now let us apply Lemma~\ref{lem:submod-reduce} to $g$ with parameter $\epsilon$. Thus, there exists a submodular function $h$ such that $\|g-h\|_2 \leq \frac12 \epsilon$, and $h$ depends only on a subset of variables $J' \subseteq J$, $|J'| \leq \frac{128}{\epsilon^2} \log \frac{16|J|}{\epsilon}$. We have $|J| \leq \frac{2^{13}}{\epsilon^5}$, and therefore
$|J'| \leq \frac{128}{\epsilon^2} \log \frac{2^{17}}{\epsilon^6} \leq \frac{128}{\epsilon^2} \log \frac{1}{\epsilon^{23}}$  (using $\epsilon \leq \frac12$). We conclude that $|J'|  \leq \frac{128 \cdot 23}{\epsilon^2} \log \frac{1}{\epsilon} \leq \frac{4000}{\epsilon^2} \log \frac{1}{\epsilon}$ as required in Theorem~\ref{thm:submod-junta}.
By the triangle inequality, we have $\|f-h\|_2 \leq \|f-g\|_2 + \|g-h\|_2 \leq \frac12 \epsilon + \frac12 \epsilon = \epsilon$. However, this would mean that the statement of Theorem~\ref{thm:submod-junta} holds for $\epsilon$ as well, which is a contradiction.
\end{proof}

In the rest of this section, our goal is to prove Lemma~\ref{lem:submod-reduce}.

\paragraph{What we need.}
Our proof relies on two previously known facts: a concentration result for submodular functions, and a ``boosting lemma" for down-monotone events.

\medskip
\noindent{\bf Concentration of submodular functions.}
It is known that a $1$-Lipschitz nonnegative submodular function $f$ is concentrated within a standard deviation of $O(\sqrt{\E[f]})$ \cite{BoucheronLM:00,Vondrak10}. This fact was also used in previous work on learning of submodular functions \cite{BalcanHarvey:12full,GuptaHRU:11,FeldmanKV:13}. Exponential tail bounds are known in this case, but we do not even need this. We quote the following result which follows from the Efron-Stein inequality (the first part is stated as Corollary 2 in \cite{BoucheronLB03}, Section 2.2; the second part follows easily from the same proof).

\begin{lemma}
\label{lem:Lipschitz}
For any self-bounding function $f:\{0,1\}^n \rightarrow \RR_+$ under a product distribution,
$$ \Var[f] \leq \E[f].$$
For any $a$-self-bounding function $f:\{0,1\}^n \rightarrow \RR_+$ under a product distribution,
$$ \Var[f] \leq a \E[f].$$
\end{lemma}

We use the fact that $1$-Lipschitz monotone submodular functions are self-bounding, and $1$-Lipschitz nonmonotone submodular functions are $2$-self-bounding (see \cite{Vondrak10}). By scaling, we obtain the following for $\alpha$-Lipschitz submodular functions (see also \cite{FeldmanKV:13}).

\begin{corollary}
\label{cor:Lipschitz}
For any $\alpha$-Lipschitz monotone submodular function $f:\{0,1\}^n \rightarrow \RR_+$ under a product distribution,
$$ \Var[f] \leq \alpha \E[f].$$
For any $\alpha$-Lipschitz (nonmonotone) submodular function $f:\{0,1\}^n \rightarrow \RR_+$ under a product distribution,
$$ \Var[f] \leq 2 \alpha \E[f].$$
\end{corollary}

\medskip
\noindent{\bf Boosting lemma for down-monotone events.}
The following was proved as Lemma 3 in \cite{GoemansVondrak:06}.

\begin{lemma}
\label{lem:boosting}
Let $\cF \subseteq \{0,1\}^X$ be down-monotone (if $x \in \cF$ and $y \leq x$ coordinate-wise, then $y \in \cF$). For $p \in (0,1)$, define
$$ \sigma_p = \Pr[X(p) \in \cF] $$
where $X(p)$ is a random subset of $X$, each element sampled independently with probability $p$. Then
$$ \sigma_p = (1-p)^{\phi(p)} $$
where $\phi(p)$ is a non-decreasing function for $p \in (0,1)$.
\end{lemma}

\paragraph{The proof of Lemma~\ref{lem:submod-reduce}}

Given a submodular function $f:\{0,1\}^J \rightarrow [0,1]$, let $F:[0,1]^J \rightarrow [0,1]$ denote the multilinear extension of $f$: $F(x) = \E[f(\hat{x})]$ where $\hat{x}$ has independently random 0/1 coordinates with expectations $x_i$. We also denote by $\b1_S$ the characteristic vector of a set $S$.

\begin{algorithm}
\label{alg:junta}
Given $f:\{0,1\}^J \rightarrow [0,1]$, produce a small set of important coordinates $J'$ as follows (for parameters $\alpha,\delta>0$):
\begin{itemize}
\item Set 
 $S = T = \emptyset$.
\item As long as there is $i \notin S$ such that $\Pr[\partial_i f(\b1_{S(\delta)}) > \alpha] > 1/2$, include $i$ in $S$. \\
{\em (This step is sufficient for monotone submodular functions.)}
\item As long as there is $i \notin T$ such that $\Pr[\partial_i f(\b1_{J \setminus T(\delta)}) < -\alpha] > 1/2$, include $i$ in $T$. \\
{\em (This step deals with non-monotone submodular functions.)}
\item Return $J' = S \cup T$.
\end{itemize}
\end{algorithm}

The intuition here (for monotone functions) is that we include greedily all variables whose contribution is significant, when measured at a random point where the variables chosen so far are set to $1$ with a (small) probability $\delta$. The reason for this is that we can bound the number of such variables, and at the same time we can prove that the contribution of unchosen variables is very small {\em with high probability}, when the variables in $J'$ are assigned uniformly at random (this part uses the boosting lemma). This is helpful in estimating the approximation error of this procedure.

First, we bound the number of variables chosen by the procedure. The argument is essentially that if the procedure had selected too many variables, their expected cumulative contribution would exceed the bounded range of the function. This argument would suffice for monotone submodular functions. The final proof is somewhat technical because of the need to deal with potentially negative discrete derivatives of non-monotone submodular functions.

\begin{lemma}
\label{lem:size}
The number of variables chosen by the procedure above is $|J'| \leq \frac{4}{\alpha \delta}$.
\end{lemma}

\begin{proof}
For each $i \in S$, let $S_{<i}$ be the subset of variables in $S$ included before the selection of $i$. For a set $R \subseteq S$ let $R_{<i}$ denote $R \cap S_{<i}$. Further, for $R \subseteq S$, let us define $R^+$ to be the set where $i \in R^+$ iff $i \in R$ and $\partial_i f(\b1_{R_{<i}}) > \alpha$; in other words, these are all the elements in $R$ that have a marginal contribution more than $\alpha$ to the previously included elements.

For each variable $i$ included in $S$, we have by definition $\Pr[\partial_i f(\b1_{S_{<i}(\delta)}) > \alpha] > 1/2$. Since each $i \in S$ appears in $S(\delta)$ with probability $\delta$, and (independently) $\partial_i f(\b1_{S_{<i}(\delta)}) > \alpha$ with probability at least $1/2$, we get that each element of $S$ appears in $S(\delta)^+$ with probability at least $\delta/2$. In expectation, $\E[|S(\delta)^+|] \geq \frac12 \delta |S|$.
Also, for any set $R \subseteq S$ and each $i \in R^+$, submodularity implies that
$\partial_i f(\b1_{R_{<i}^+}) \geq \partial_i f(\b1_{S_{<i}}) > \alpha$, since $R_{<i}^+ \subseteq R_{<i} \subseteq S_{<i}$.
Now we get that $$f(R^+) = f({\bf 0}) + \sum_{i\in R^+} \partial_i f(\b1_{R_{<i}^+}) > \alpha |R^+|.$$
From here we obtain that $$\E[f(S(\delta)^+)] > \alpha \E[|S(\delta)^+|] \geq \frac12 \alpha \delta |S|.$$ This implies that $|S| \leq \frac{2}{\alpha \delta}$, otherwise the expectation would exceed the range of $f$, which is $[0,1]$.

To bound the size of $T$ we observe that the function $\bar{f}$ defined as $\bar{f}(\b1_R) = f(\b1_{J\setminus R})$ for every $R \subseteq J$ is submodular and for every $i \in J$, $\partial_i \bar{f}(\b1_R) = - \partial_i f(\b1_{J \setminus R})$. The criterion for including the variables in $T$ is the same as criterion of including the variables in $S$ used for function $\bar{f}$ in place of $f$. Therefore, by an analogous argument, we cannot include more than $\frac{2}{\alpha \delta}$ elements in $T$, hence $|J'| = |S \cup T| \leq \frac{4}{\alpha \delta}$.
\end{proof}

The next step in the analysis replaces the condition used by Algorithm~\ref{alg:junta} by a probability bound exponentially small in $1/\delta$. The tool that we use here is the ``boosting lemma" (Lemma~\ref{lem:boosting}) which amplifies the probability bound from $1/2$ to $1/2^{1/(2\delta)}$, as the sampling probability goes from $\delta$ to $1/2$.

\begin{lemma}
\label{lem:high-prob}
With the same notation as above, if $\delta \leq 1/2$, then for any $i \in J \setminus J'$
$$ \Pr[\partial_i f(\b1_{J'(1/2)}) > \alpha] \leq 2^{-1/(2\delta)} $$
and
$$ \Pr[\partial_i f(\b1_{J \setminus J'(1/2)}) < -\alpha] \leq 2^{-1/(2\delta)}. $$
\end{lemma}

\begin{proof}
Let us prove the first inequality; the second one will be similar.
First, we know by the selection rule of the algorithm that for any $i \notin J'$,
$$ \Pr[\partial_i f(\b1_{S(\delta)}) > \alpha] \leq 1/2. $$
By submodularity of $f$ we get that for any $i \notin J'$,
$$ \Pr[\partial_i f(\b1_{J'(\delta)}) > \alpha] \leq 1/2. $$
Denote by $\cF \subseteq \{0,1\}^{J'}$ the family of points $x$ such that $\partial_i f(x) > \alpha$. By the submodularity of $f$, which is equivalent to partial derivatives being non-increasing, $\cF$ is a down-monotone set: if $y \leq x \in \cF$, then $y \in \cF$. If we define $\sigma_p = \Pr[J'(p) \in \cF]$ as in Lemma~\ref{lem:boosting}, we have $\sigma_\delta \leq 1/2$. Therefore, by Lemma~\ref{lem:boosting}, $\sigma_p = (1-p)^{\phi(p)}$ where $\phi(p)$ is a non-decreasing function. For $p = \delta$, we get $\sigma_\delta = (1-\delta)^{\phi(\delta)} \leq 1/2$, which implies $\phi(\delta) \geq 1/(2\delta)$ (note that $(1-\delta)^{1/(2\delta)} \geq 1/2$ for any $\delta \in [0,1/2]$). As $\phi(p)$ is non-decreasing, we must also have $\phi(1/2) \geq 1/(2\delta)$. This means $\sigma_{1/2} = (1/2)^{\phi(1/2)} \leq 1/2^{1/(2\delta)}$. Recall that $\sigma_{1/2} = \Pr[J'(1/2) \in \cF] = \Pr[\partial_i f(\b1_{J'(p)}) > \alpha]$ so this proves the first inequality.

For the second inequality, we denote similarly $\cF' = \{ F \subseteq J': \partial_i f(\b1_{J \setminus F}) < -\alpha \}$.
Again, this is a down-monotone set by the submodularity of $f$.
By the selection rule of the algorithm,
$ \sigma'_\delta = \Pr[J'(\delta) \in \cF'] = \Pr[\partial_i f(\b1_{J\setminus J'(\delta)}) < -\alpha] \leq
\Pr[\partial_i f(\b1_{J \setminus T(\delta)}) < -\alpha] \leq 1/2.$
This implies by Lemma~\ref{lem:boosting} that $\sigma'_{1/2} = \Pr[J'(1/2) \in \cF'] \leq 1/2^{1/(2\delta)}$. This proves the second inequality.
\end{proof}

\begin{proof}[Proof of Lemma~\ref{lem:submod-reduce}]
Given a submodular function $f:\{0,1\}^J \rightarrow [0,1]$, we construct a set of coordinates $J' \subseteq J$ as described above, with parameters $\alpha = \frac{1}{16} \epsilon^2$ and $\delta = 1 / (2\log \frac{16|J|}{\epsilon^2})$. Lemma~\ref{lem:size} guarantees that $|J'| \leq \frac{4}{\alpha \delta} = \frac{128}{\epsilon^2} \log \frac{16|J|}{\epsilon^2}$.

Let us use $x_{J'}$ to denote the $|J'|$-tuple of coordinates of $x$ indexed by $J'$.
Consider the subcube of $\{0,1\}^J$ where the coordinates on $J'$ are fixed to be $x_{J'}$.
In the following, all expectations are over a uniform distribution on the respective subcube, unless otherwise indicated.
We denote by $f_{x_{J'}}$ the restriction of $f$ to this subcube, $f_{x_{J'}}(y) = f(x_{J'},y)$.
We define $h:\{0,1\}^J \rightarrow [0,1]$ to be the function obtained by replacing each $f_{x_{J'}}$ by its expectation over the respective subcube:
$$ h(x) = \E[f_{x_{J'}}] = \E_{y \in \{0,1\}^{\bar{J'}}}[f(x_{J'},y)].$$
Obviously $h$ depends only on the variables in $J'$ and it is easy to see that it is submodular with range in $[0,1]$.
It remains to estimate the distance of $h$ from $f$. Observe that
\begin{eqnarray*}
 \|f - h\|_2^2 & = & \E_{x \in \{0,1\}^J}[(f(x) - h(x))^2] \\
  & = & \E_{x_{J'} \in \{0,1\}^{J'}} \E_{y \in \{0,1\}^{\bar{J'}}}[(f(x_{J'},y) - h(x_{J'},y))^2] \\
  & = & \E_{x_{J'} \in \{0,1\}^{J'}} \E_{y \in \{0,1\}^{\bar{J'}}}[(f_{x_{J'}}(y) - \E[f_{x_{J'}}])^2] \\
  & = & \E_{x_{J'} \in \{0,1\}^{J'}} [\Var[f_{x_{J'}}]].
\end{eqnarray*}
We partition the points $x_{J'} \in \{0,1\}^{J'}$ into two classes:
\begin{enumerate}
\item Call $x_{J'}$ bad, if there is $i \in J \setminus J'$ such that
\begin{itemize}
\item  $ \partial_i f(x_{J'}) > \alpha$, or
\item  $ \partial_i f(x_{J'} + \b1_{J \setminus J'}) < -\alpha.$
\end{itemize}
In particular, we call $x_{J'}$ bad for the coordinate $i$ where this happens.
\item Call $x_{J'}$ good otherwise, i.e. for every $i \in J \setminus J'$ we have
\begin{itemize}
\item  $ \partial_i f(x_{J'}) \leq \alpha $, and
\item $ \partial_i f(x_{J'} + \b1_{J \setminus J'}) \geq -\alpha.$
\end{itemize}
\end{enumerate}
Consider a good point $x_{J'}$ and the restriction of $f$ to the respective subcube, $f_{x_{J'}}$. The condition above means that for every $i \in J \setminus J'$, the marginal value of $i$ is at most $\alpha$ at the bottom of this subcube, and at least $-\alpha$ at the top of this subcube. By submodularity, it means that the marginal values are between $[-\alpha, \alpha]$, for all points of this subcube. Hence, $f_{x_{J'}}$ is a $\alpha$-Lipschitz submodular function. By Corollary~\ref{cor:Lipschitz},
$$ \Var[f_{x_{J'}}] \leq 2 \alpha \E[f_{x_{J'}}] \leq \frac{1}{8} \epsilon^2 $$
considering that $\alpha = \frac{1}{16} \epsilon^2$ and $f_{x_{J'}}$ has values in $[0,1]$.

If $x_{J'}$ is bad, then we do not have a good bound on the variance of $f_{x_{J'}}$. However, there cannot be too many bad points $x_{J'}$, due to Lemma~\ref{lem:high-prob}: Observe that the distribution of $x_{J'}$, uniform in $\{0,1\}^{J'}$,  is the same as what we denoted by $\b1_{J'(1/2)}$ in Lemma~\ref{lem:high-prob}, and the distribution of $x_{J'} + \b1_{J \setminus J'}$ is the same as $\b1_{J \setminus J'(1/2)}$. By Lemma~\ref{lem:high-prob}, we have that for each $i \in J \setminus J'$, the probability that $x_{J'}$ is bad for $i$ is at most $2 \cdot 2^{1/(2\delta)} = \frac{\epsilon^2}{8|J|}$. By a union bound over all coordinates $i \in J \setminus J'$, the probability that $x_{J'}$ is bad is at most $\frac18 \epsilon^2$.

Now we can estimate the $\ell_2$-distance between $f$ and $h$:
\begin{eqnarray*}
 \|f - h\|_2^2 & = & \E_{x_{J'} \in \{0,1\}^{J'}} [\Var[f_{x_{J'}}]] \\
  & \leq & \Pr[x_{J'} \mbox{ is bad}] \cdot 1 + \Pr[x_{J'} \mbox{ is good}] \cdot
  \E_{\mbox{\small good } x_{J'}} [\Var[f_{x_{J'}}]] \\
  & \leq & \Pr[x_{J'} \mbox{ is bad}] + \max_{\mbox{\small good } x_{J'} } [\Var[f_{x_{J'}}]]\\
  & \leq & \frac18 \epsilon^2 + \frac18 \epsilon^2 = \frac14 \epsilon^2.
\end{eqnarray*}
Hence, we conclude that $\|f-h\|_2 \leq \frac12 \epsilon$ as desired.
\end{proof}

We now briefly examine the special case of a submodular function taking values in $\{0,\fr{k},\frac{2}{k},\ldots,1\}$ for some integer $k$. This is just a scaled version of the pseudo-boolean case considered in \cite{RaskhodnikovaYaroslavtsev:13} and  \cite{BlaisOSY:13manu}.
By choosing $\alpha = \fr{k+1}$ and $\delta = 1/(2\log{\frac{2|J|}{\eps}})$ in the proof above we will obtain that an $\alpha$-Lipschitz function must be a constant (and, in particular, independent of all the variables in $J\setminus J'$). This means that we obtain exact equality for all but the ``bad" values of $x_{J'}$. The fraction of such values is at most $2 \cdot 2^{1/(2\delta)} \cdot |J| \leq \eps$ and therefore the submodular function $h(x) = f(x_J,\b1_{J\setminus J'})$ equals $f$ with probability at least $1-\epsilon$. As before, after one application we get a $O(k \cdot \log{(n/\eps)})$-junta and by repeating the application we can obtain a $O(k \cdot \log{(k/\eps)})$-junta.
\begin{corollary}
\label{cor:submod-junta-pseudo}
For any integer $k \geq 1$, $\epsilon \in (0,\frac12)$ and any submodular function $f:\{0,1\}^n \rightarrow \{0,1,\ldots,k\}$,
there exists a submodular function $g:\{0,1\}^n \rightarrow \{0,1,\ldots,k\}$ depending only on a subset of variables $J \subseteq [n]$, $|J| = O(k \log \frac{k}{\epsilon})$, such that $\Pr_\U[f \neq g] \leq \epsilon$.
\end{corollary}

\subsection{Multiplicative Approximation for Monotone Submodular Functions}
\label{sec:PMAC-junta}

In this section we show how our approximation theorem can be extended to multiplicative approximation with high probability as required by the PMAC model, introduced by Balcan and Harvey \cite{BalcanHarvey:12full}.
We prove that  for any $\gamma>1, \epsilon>0$, a multiplicative $(\gamma,\epsilon)$-approximation for monotone submodular functions over the uniform distribution can be achieved by a function $h$ of a subset of variables whose cardinality depends only on $\gamma$ and $\epsilon$. More precisely, we prove the following.

\begin{theorem}[restatement of Theorem~\ref{thm:PMAC-junta}]
For every monotone submodular function $f:\{0,1\}^n \rightarrow \RR_+$ and every $\gamma, \epsilon \in (0,1)$, there is a monotone submodular function $h:\{0,1\}^J \rightarrow \RR_+$ depending only on a subset of variables $J \subseteq [n], |J| \leq \frac{2^{12}}{\gamma^2} \log \frac{16}{\gamma \epsilon} \log \frac{4}{\epsilon}$ such that $h$ is a multiplicative $(1+\gamma,\epsilon)$-approximation of $f$ over the uniform distribution. The function $h$ can be found with high probability using $\poly(n)$ value queries to $f$.
\end{theorem}

Observe that (for monotone submodular functions and ignoring the additional logarithm) this is stronger than Theorem~\ref{thm:submod-junta}: For any function with range $[0,1]$, a multiplicative $(1+\epsilon,\epsilon)$-approximation implies an additive error bounded by $\epsilon$, except for probability measure of $\epsilon$, which means the $\ell_1$ error is bounded by $2 \epsilon$.

The proof of Theorem~\ref{thm:PMAC-junta} is algorithmic and uses several ideas from the proof of Theorem~\ref{thm:submod-junta}. Again, we rely on the boosting lemma and concentration of submodular functions. However, the requirement of a multiplicative approximation to the target function leads to additional complications that we have been able to resolve only in the case of monotone submodular functions. We are not sure whether the theorem holds for non-monotone submodular functions, which we leave as an open question.


As in the case of $\ell_2$-error, to prove Theorem~\ref{thm:PMAC-junta} it is sufficient to prove the following statement.

\begin{lemma}
\label{lem:PMAC-reduce}
For every monotone submodular function $f:\{0,1\}^J \rightarrow \RR_+$ and every $\gamma, \epsilon \in (0,1)$, there is a monotone submodular function $h:\{0,1\}^J \rightarrow \RR_+$ depending only on a subset of variables $J' \subseteq J, |J'| \leq \frac{2^9}{\gamma^2} \log \frac{4}{\epsilon} \log \frac{2|J|}{\epsilon}$ such that $h$ is a multiplicative $(1+\gamma,\epsilon)$-approximation of $f$ over the uniform distribution.
\end{lemma}

We find the desired set of significant variables by the following procedure, a modification of Algorithm~\ref{alg:junta}.

\begin{algorithm}
\label{alg:PMAC-junta}
Given $f:\{0,1\}^J \rightarrow [0,1]$, produce $J' \subseteq J$ as follows (for parameters $\beta,\delta>0$):
\begin{itemize}
\item Set $S := \emptyset$. We use $S(\delta)$ to denote a random subset of $S$ where each element appears independently with probability $\delta$.
\item As long as there is $i \notin S$ such that
$$ \Pr_{T \sim S(\delta)}[\partial_i f(\b1_T) > \beta f(\b1_{T \cup \bar{S}})] > \frac12 $$
include $i$ in $S$ and repeat. \\
\item Return $J' := S$.
\end{itemize}
\end{algorithm}

The intuition here is that variables get included in $S$ based on their contribution relative to $f(\b1_{T \cup \bar{S}}) = f(\b1_{S(\delta) \cup \bar{S}})$. Note that this is the top of the subcube defined by fixing the coordinates on $S$ to be equal to $x_S = \b1_T$. This is important for obtaining a decomposition such that in each such subcube, the function is sufficiently smooth relative to its own expectation and hence approximated by a constant within a small multiplicative factor. On the other hand, we can bound the number of variables that can be included in $S$ as follows.

\begin{lemma}
\label{lem:PMAC-size}
The cardinality of the set $J'$ returned by Algorithm~\ref{alg:PMAC-junta} is at most $2/(\beta \delta)$.
\end{lemma}

\begin{proof}
Consider the ordering of elements as they were selected by the algorithm, and assume w.l.o.g. that the ordering is $\{1,2,3,\ldots,|J'|\}$. Whenever an element $i$ is included, it is because $\Pr_{T \sim S(\delta)}[\partial_i f(\b1_T) > \beta f(\b1_{T \cup \bar{S}})] > \frac12$. Here, $S$ is the set of elements selected before $i$, that is $S = [i-1]$ in our ordering. Thus we can write $T = S(\delta) = R \cap [i-1]$, where $R = J'(\delta)$. The condition above can be written as $\Pr_{R \sim J'(\delta)}[\partial_i f(\b1_{R \cap [i-1]}) > \beta f(\b1_{R  \cup \overline{[i-1]}})] > \frac12$. For each $R \subseteq J'$, let us define $R^+$ as
$$R^+ = \{i \in R: \partial_i f(\b1_{R \cap [i-1]}) > \beta f(\b1_{R \cup \overline{[i-1]}}) \}.$$
Observe that by a telescoping sum, $$f(\b1_R) = f({\bf 0}) + \sum_{i \in R} \partial_i f(\b1_{R \cap [i-1]}) > \beta \sum_{i \in R^+} \partial_i f(\b1_{R \cup \overline{[i-1]}}) \geq |R^+| \cdot \beta f(\b1_R)$$ and hence $|R^+| < 1/\beta$ for every $R$.

Consider the expectation $\E_{R \sim J'(\delta)}[|R^+|]$. As we argued above, every time we include $i$ in $J'$, we have the property that $\Pr_{R \sim J'(\delta)}[f(\b1_{R \cap [i-1]}) > \beta f(\b1_{R \cup \overline{[i-1]}})] > \frac12$. Since $i$ appears in $R$ with probability $\delta$, independently of the condition $\partial_i f(\b1_{R \cap [i-1]}) > \beta f(\b1_{R \cup \overline{[i-1]}})$, this means that each element $i \in J'$ appears in $R^+$ with probability at least $\delta/2$. We conclude that $\E_{R \sim J'(\delta)}[|R^+|] \geq |J'| \delta / 2$. On the other hand, $|R^+| < 1/\beta$ for all $R$. This implies that $|J'| < 2 / (\beta \delta)$.
\end{proof}

Recall that so far, we were working with subsets of $J'$ sampled with a (small) probability $\delta$.
The next step is to prove that for a {\em uniformly} random assignment $x_{J'} \in \{0,1\}^{J'}$,  the function $f_{x_{J'}}(y) = f(x_{J'},y)$ for $y \in \{0,1\}^{\bar{J'}}$ has suitable Lipschitz properties for most values of $x_{J'}$. This relies on the boosting lemma, and in this step we require again that $f$ is a {\em monotone} submodular function.
In the following, all expectations are over a uniform distribution on the respective subcube, unless otherwise indicated.

\begin{lemma}
\label{lem:PMAC-boost}
The set $J'$ returned by Algorithm~\ref{alg:PMAC-junta} satisfies for every $i \notin J'$,
$$ \Pr_{x_{J'} \in \{0,1\}^{J'}}[\partial_i f(x_{J'}) > \beta f(x_{J'},\b1_{\bar{J'}})] \leq 2^{-1/(2\delta)}.$$
\end{lemma}

\begin{proof}
Denote by $\cF \subseteq \{0,1\}^{J'}$ the family of points $x_{J'}$ such that the condition is satisfied, i.e.
$\cF = \{ x_{J'} \in \{0,1\}^{J'}: \partial_i f(x_{J'}) > \beta f(x_{J'},\b1_{\bar{J'}}) \}$. This is a down-monotone set: if $y \leq x  \in \cF$, then $y \in \cF$ because $\partial_i f(y) \geq \partial_i f(x)$, and $f(y,\b1_{\bar{J'}}) \leq f(x,\b1_{\bar{J'}})$ (here we are using both monotonicity and submodularity).

If we define $\sigma_p = \Pr[\b1_{J'(p)} \in \cF]$, this means that $\sigma_\delta \leq 1/2$. By Lemma~\ref{lem:boosting}, we have $\Pr[\b1_{J'(1/2)} \in \cF] = \sigma_{1/2} \leq 2^{-1/(2\delta)}$. As $\b1_{J'(1/2)}$ is distributed uniformly in $\{0,1\}^{J'}$, this is exactly the statement of Lemma~\ref{lem:PMAC-boost}.
\end{proof}

Finally, we finish the proof of Lemma~\ref{lem:PMAC-reduce} by using concentration properties of submodular functions. We refer to the following bound from \cite{Vondrak10}.

\begin{lemma}
\label{lem:submod-chernoff}
If $Z = f(X_1,\ldots,X_n)$ where $X_i \in \{0,1\}$ are independently random and $f$ is a nonnegative submodular function with discrete derivatives bounded by $[-1,1]$, then for any $\lambda>0$,
\begin{itemize}
\item $\Pr[Z \geq (1+\lambda) \E[Z]] \leq e^{-\lambda^2 \E[Z] / (4+5\lambda/3)}$.
\item $\Pr[Z \leq (1-\lambda) \E[Z]] \leq e^{-\lambda^2 \E[Z]/4}$.
\end{itemize}
\end{lemma}

\begin{proof}[Proof of Lemma~\ref{lem:PMAC-reduce}]
Given a monotone submodular function $f:\{0,1\}^J \rightarrow \RR_+$ and $\gamma,\epsilon \in (0,1)$, we construct $J' \subseteq J$ by running Algorithm~\ref{alg:PMAC-junta} with parameters $\beta = \frac{1}{108} \gamma^2 / \log \frac{4}{\epsilon}$ and $\delta = 1 / (2 \log \frac{2|J|}{\epsilon})$. By Lemma~\ref{lem:PMAC-size}, the constructed subset of variables has size $|J'| \leq 2 / (\beta \delta) \leq 2^9 \gamma^{-2} \log \frac{4}{\epsilon} \log \frac{2|J|}{\epsilon}$.

By Lemma~\ref{lem:PMAC-boost}, we obtain a subset of variables $J'$ such that for every $i \notin J'$,
$$ \Pr_{x_{J'} \in \{0,1\}^{J'}}[\partial_i f(x_{J'}) > \beta f(x_{J'},\b1_{\bar{J'}})] \leq \frac{\epsilon}{2|J|}.$$
By the union bound,
$$ \Pr_{x_{J'} \in \{0,1\}^{J'}}[\exists i \in J \setminus J'; \partial_i f(x_{J'}) > \beta f(x_{J'},\b1_{\bar{J'}})] \leq \frac{\epsilon}{2}.$$
This means that with probability $1-\epsilon/2$ over the choice of $x_{J'} \in \{0,1\}^{J'}$, the point $x_{J'}$ is {\em good} in the sense that the function $f_{x_{J'}}(y) = f(x_{J'},y)$ for $y \in \{0,1\}^{\bar{J'}}$ has discrete derivatives bounded by $\partial_i f(x_{J'}) \leq \beta f_{x_{J'}}(\b1_{\bar{J'}})$. Fix any good point $x_{J'}$. By submodularity, the same bound holds for the derivatives evaluated at any point above $x_{J'}$. In addition, $f$ is monotone, hence $\partial_i f_{x_{J'}}(y) \in [0, \beta f_{x_{J'}}(\b1_{\bar{J'}})]$ for all $y \in \{0,1\}^{\bar{J'}}$.

Here we use a concentration bound for submodular functions (Lemma~\ref{lem:submod-chernoff}). Consider the function $f_{x_{J'}}$ for a good point $x_{J'}$. We apply the concentration bound to a scaled function $\tilde{f}(y) = f_{x_{J'}}(y) / (\beta f_{x_{J'}}(\b1_{\bar{J'}}))$. By the discussion above, $\tilde{f}$ has discrete derivatives in $[0,1]$. By Lemma~\ref{lem:submod-chernoff}, for $\lambda \in [0,1]$,
$$ \Pr_{y \in \{0,1\}^{\bar{J'}}}[|\tilde{f}(y) - \E[\tilde{f}]| > \lambda \E[\tilde{f}]] < 2 e^{-\lambda^2 \E[\tilde{f}] / 6}.$$
We also use a known fact \cite{Feige:06} that for any monotone submodular function, $\E[\tilde{f}] \geq \frac12 \|\tilde{f}\|_\infty = \frac12 \tilde{f}(\b1_{\bar{J'}}) = 1 / (2\beta)$. Going back to $f_{x_{J'}}$, we obtain
$$ \Pr_{y \in \{0,1\}^{\bar{J'}}}[|f_{x_{J'}}(y) - \E[f_{x_{J'}}]| > \lambda \E[f_{x_{J'}}]] < 2 e^{-\lambda^2 / (12 \beta)}.$$
We set $\lambda = \gamma/3$, and recall that we have $\beta = \frac{1}{108} \gamma^2 / \log \frac{4}{\epsilon}$. Therefore
$$ \Pr_{y \in \{0,1\}^{\bar{J'}}}[|f_{x_{J'}}(y) - \E[f_{x_{J'}}]| > \frac13 \gamma \E[f_{x_{J'}}]] < 2 e^{-\log \frac{4}{\epsilon}} \leq \frac{\epsilon}{2} $$
for every good point $x_{J'}$. Equivalently,
\begin{equation}
\label{eq:good-chernoff}
\Pr_{y \in \{0,1\}^{\bar{J'}}} \left[\frac{f_{x_{J'}}(y)}{1+\gamma/3} \leq \E[f_{x_{J'}}] \leq \frac{f_{x_{J'}}(y)}{1-\gamma/3} \right]
 > 1 - \frac{\epsilon}{2}
\end{equation}
for every good point $x_{J'}$.

We define our approximation to $f$ as follows:
$$ h(x) = \left(1 + \frac{\gamma}{3} \right) \E[f_{x_{J'}}]. $$
In other words, we average out the contributions of all variables outside of $J'$, and we adjust by a constant factor of $1+\frac{\gamma}{3}$, to make sure that $h(x) \geq f(x)$ with the desired probability. Observe that $h$ is a positive linear combination of monotone submodular functions, and hence also a monotone submodular function. Also, $h$ depends only on the variables in $J'$.

Now, our goal is to estimate the probability that $f(x) \leq h(x) \leq (1+\gamma) f(x)$. In the following, all probabilities and expectations are over uniform distributions. We have
\begin{eqnarray*}
& & \Pr_{x \in \{0,1\}^J}[f(x) \leq h(x) \leq (1+\gamma) f(x)] \\
 & = & \E_{x_{J'} \in \{0,1\}^{J'}}\left[ \Pr_{y \in \{0,1\}^{\bar{J'}}}[f_{x_{J'}}(y) \leq h(x_{J'}) \leq (1+\gamma) f_{x_{J'}}(y)] \right] \\
 & = & \E_{x_{J'} \in \{0,1\}^{J'}}\left[ \Pr_{y \in \{0,1\}^{\bar{J'}}}[f_{x_{J'}}(y) \leq (1+\gamma/3) \E[f_{x_{J'}}] \leq (1+\gamma) f_{x_{J'}}(y)] \right] \\
 & \geq & \E_{x_{J'} \in \{0,1\}^{J'}} \left[ \Pr_{y \in \{0,1\}^{\bar{J'}}}\left[\frac{f_{x_{J'}}(y)}{1+\gamma/3} \leq \E[f_{x_{J'}}] \leq \frac{f_{x_{J'}}(y)}{1-\gamma/3} \right]\right] \\
 & \geq & \Pr_{x_{J'} \in \{0,1\}^{J'}}[x_{J'} \mbox{ is good}] \cdot \left( 1 - \frac{\epsilon}{2} \right)
\end{eqnarray*}
using Eq.~(\ref{eq:good-chernoff}). As we argued above, a uniformly random point $x_{J'} \in \{0,1\}^{J'}$ is good with probability at least $1 - \epsilon/2$. Hence,
$$ \Pr_{x \in \{0,1\}^J}[f(x) \leq h(x) \leq (1+\gamma) f(x)] \geq \left(1 - \frac{\epsilon}{2} \right)^2 \geq 1 - \epsilon $$
which is the definition of multiplicative $(1+\gamma,\epsilon)$-approximation.
\end{proof}

Now we can finish the proof of Theorem~\ref{thm:PMAC-junta}.

\begin{proof}[Proof of Theorem~\ref{thm:PMAC-junta}]
Given a monotone submodular function $f:\{0,1\}^n \rightarrow \RR_+$, we use Lemma~\ref{lem:PMAC-reduce} repeatedly to reduce the number of variables. To work out the necessary parameters, we proceed backwards: Eventually, we want to obtain a multiplicative $(1+\gamma,\epsilon)$-approximation, using $O(\frac{1}{\gamma^2} \log \frac{1}{\gamma \epsilon} \log \frac{1}{\epsilon})$ variables. Let us define the following sequences: (for a constant $c$ to be determined later)
\begin{itemize}
\item $\gamma_i = \gamma / 2^{i}$, $\epsilon_i = \epsilon / 2^{i}$,
\item $n_0 = \lfloor \frac{c}{\gamma^2} \log \frac{16}{\gamma \epsilon} \log \frac{4}{\epsilon} \rfloor$,
\item $n_{i+1} = \lfloor \frac{\epsilon_i}{2} \cdot 2^{n_i \gamma_i^2 / (2^{9} \log \frac{4}{\epsilon_i})} \rfloor$.
\end{itemize}
The meaning of this sequence is that given a function $f_{i+1}$ of $n_{i+1}$ variables and parameters $\gamma_i, \epsilon_i$, we can find a function $f_i$ of $n_i$ variables which is a multiplicative $(1+\gamma_i, \epsilon_i)$-approximation of $f_{i+1}$ (using Lemma~\ref{lem:PMAC-reduce}, and inverting the relationship between $n_i$ and $n_{i+1}$). Note that the parameters $\gamma_i, \epsilon_i$ form geometric series adding up to at most $\gamma$ and $\epsilon$ respectively, and consequently $f_0$ is a multiplicative $(1+\gamma,\epsilon)$-approximation of $f_k$ for any $k>0$.

By induction, we prove the following for every $i \geq 0$:
\begin{equation}
\label{eq:PMAC-induc}
n_i \geq \lfloor c \frac{1}{\gamma_i^2} \log \frac{16}{\gamma_i \epsilon_i} \log \frac{4}{\epsilon_i} \rfloor.
\end{equation}
The base case holds by definition. So assume that (\ref{eq:PMAC-induc}) holds for $n_i$. For $n_{i+1}$, we obtain
\begin{eqnarray*}
n_{i+1} & = & \lfloor \frac{\epsilon_i}{2} \cdot 2^{n_i \gamma_i^2 / (2^{9} \log \frac{4}{\epsilon_i})} \rfloor  \\
& \geq & \lfloor \frac{\epsilon_i}{2} \cdot 2^{c (\log \frac{16}{\gamma_i \epsilon_i}) / 2^{10}} \rfloor \\
& = & \left\lfloor \frac{\epsilon_i}{2} \cdot \left(\frac{16}{\gamma_i \epsilon_i} \right)^{c / 2^{10}} \right\rfloor.
\end{eqnarray*}
We pick $c = 2^{14}$, and use $\gamma_i = 2 \gamma_{i+1},  \epsilon_i = 2 \epsilon_{i+1}$, which yields
\begin{eqnarray*}
n_{i+1} & \geq &
 \left\lfloor \epsilon_{i+1} \cdot \left(\frac{4}{\gamma_{i+1} \epsilon_{i+1}} \right)^{16} \right\rfloor \\
 & = & \lfloor 2^{32} \gamma_{i+1}^{-16} \epsilon_{i+1}^{-15} \rfloor \\
 & \geq & \left\lfloor  \frac{2^{14}}{\gamma_{i+1}^2} \log \frac{16}{\gamma_{i+1} \epsilon_{i+1}} \log \frac{4}{\epsilon_{i+1}} \right\rfloor.
\end{eqnarray*}
This proves Eq.~(\ref{eq:PMAC-induc}). Note that in particular, since $\gamma_i = \gamma/2^{i}, \epsilon_i = \epsilon/2^{i}$, this proves that $n_i$ grows at least as a geometric sequence, and will reach $n_t \geq n$ in $t = O(\log n)$ steps (in fact much faster, but we are not concerned with the exact number of iterations). Therefore, we can take $f_t = f$ to be our original function and work our way backwards, to obtain a multiplicative $(1+\gamma,\epsilon)$-approximation $f_0$ which depends on at most $n_0 = \lfloor  \frac{2^{14}}{\gamma^2} \log \frac{16}{\gamma \epsilon} \log \frac{4}{\epsilon} \rfloor$ variables.

We remark that the proof is constructive and we have in fact constructed the multiplicative junta approximation by a randomized polynomial-time algorithm (with value query access to $f$) that succeeds with high probability.
\end{proof}

\section{Approximation of Low-Influence Functions by Juntas}
\label{sec:low-sensitivity}

Here we show how structural results for submodular (weaker than the one in Section~\ref{sec:submod-junta}), XOS and self-bounding functions  can be proved in a unified manner using the notion of total influence.

\subsection{Preliminaries: Fourier Analysis}
\label{sec:lowsens-prelims}
We rely on the standard Fourier transform representation of real-valued functions over $\zon$ as linear combinations of parity functions.
For $S \subseteq [n]$, the parity function $\chi_S:\zon \rightarrow \on$ is defined by
$ \chi_S(x) = (-1)^{\sum_{i \in S} x_i}$.
\eat{ 
The parities form an orthonormal basis for functions on $\zon$ under the inner product product with respect to the uniform distribution. Thus, every function $f: \zon \rightarrow \RR$ can be written as a real linear combination of parities. The coefficients of the linear combination are referred to as Fourier coefficients of $f$.
For $f:\zon \rightarrow \RR$ and $S \subseteq [n]$, the Fourier coefficient $\hat{f}(S)$ is given by $\hat{f}(S) = \langle f, \chi_S \rangle.$ For any Fourier coefficient $\hat{f}(S)$, $|S|$ is called the \emph{degree} of the coefficient.
}
The Fourier expansion of $f$ is given by $f(x) = \sum_{S \subseteq [n]} \hat{f}(S) \chi_S(x)$. The {\em Fourier degree} of $f$ is the largest $|S|$ such that $\hat{f}(S) \neq 0$.
Note that Fourier degree of $f$ is exactly the polynomial degree of $f$ when viewed over $\on^n$ instead of $\zon$ and therefore it is also equal to the polynomial degree of $f$ over $\zon$. Let $f: \zon \rightarrow \RR$ and $\hat{f}: 2^{[n]} \rightarrow \RR$ be its Fourier transform. The {\em spectral $\ell_1$-norm} of $f$ is defined as $ \|\hat{f}\|_1 = \sum_{S \subseteq [n]} |\hat{f}(S)|.$

Observe that
$\partial_i f(x) = -2 \sum_{S \ni i}\hat{f}(S)\chi_{S\setminus\{i\}}(x)$, and
$\partial_{i,j} f(x) = 4 \sum_{S \ni i,j}\hat{f}(S)\chi_{S\setminus\{i,j\}}(x)$.

We use several notions of {\em influence} of a variable on a real-valued function which are based on the standard notion of influence for Boolean functions (\eg \cite{Ben-OrLinial:85,KahnKL:88}).
\begin{definition}[Influences]
For a real-valued $f:\zo^n \rightarrow \RR$, $i \in [n]$, and $\kappa \geq 0$ we define the {\em $\ell_\kappa^\kappa$-influence} of variable $i$ as $\infl^\kappa_i(f) = \|\fr{2}\partial_i f\|_\kappa^\kappa = \E[|\fr{2}\partial_i f|^\kappa]$. We define $\infl^\kappa(f) = \sum_{i\in[n]} \infl^\kappa_i(f)$ and refer to it as the {\em total $\ell_\kappa^\kappa$-influence} of $f$. For a boolean function $f:\zon\rightarrow \zo$, $\infl(f)$ is defined as $2\cdot \infl^1(f)$ and is also referred to as {\em average sensitivity}.
\end{definition}

The most commonly used notion of influence for real-valued functions is the $\ell_2^2$-influence which satisfies
$$\infl^2_i(f) = \left\|\fr{2}\partial_i f\right\|_2^2 = \sum_{S \ni i}\hat{f}^2(S)\ .$$
From here, the total $\ell_2^2$-influence is equal to $\infl^2(f) = \sum_S |S| \hat{f}^2(S)$.

\subsection{Self-bounding Functions Have Low Total Influence}
A key fact that we prove is that submodular, XOS and self-bounding functions have low total $\ell_1$-influence.

\begin{lemma}
\label{lem:submod}
Let $f:\zo^n\rightarrow \RR_+$ be an $a$-self-bounding function. Then $\infl^1(f) \leq a\cdot \|f\|_1$.
In particular, for an XOS function $f:\zo^n\rightarrow [0,1]$, $\infl^1(f) \leq 1$.
For a submodular $f:\zo^n \rightarrow [0,1]$, $\infl^1(f) \leq 2$.
\end{lemma}
\begin{proof}
We have
$$ \infl^1(f) = \frac12 \sum_{i=1}^{n} \E[|f(x_{i \leftarrow 1}) - f(x_{i \leftarrow 0})|]
 = \sum_{i=1}^{n} \E[(f(x) - f(x \oplus e_i))_+] $$
where $x \oplus e_i$ is $x$ with the $i$-th bit flipped, and $(\bullet)_+ = \max \{\bullet,0\}$
is the positive part of a number. (Note that each difference $|f(x_{i \leftarrow 1}) - f(x_{i \leftarrow 0})|$ is counted twice in the first expectation and once in the second expectation.) By using the property of $a$-self-bounding functions, we know that
$ \sum_{i=1}^{n} (f(x) - f(x \oplus e_i))_+ \leq a f(x) $, which implies
$$ \infl^1(f) = \sum_{i=1}^{n} \E[(f(x) - f(x \oplus e_i))_+] \leq a \E[|f(x)|] = a \|f\|_1.$$
Finally, we recall that an XOS function is self-bounding and a non-negative submodular function is 2-self-bounding (see \cite{Vondrak10}).
\end{proof}

We note that for functions with a $[0,1]$ range, $\infl^2 (f) \leq \infl^1 (f)$, hence the above lemma also gives a bound on $\infl^2(f)$.
It is well-known that functions of low total $\ell_2^2$-influence can be approximated by low-degree polynomials. We recap this fact here.

\begin{lemma}
\label{lem:low-degree-conc}
Let $f:\zo^n\rightarrow \RR$ be any function and let $d$ be any positive integer. Then $\sum_{ S \subseteq[n], |S| > d } \hat{f}(S)^2 \leq \infl^2(f)/d $.
\end{lemma}
\begin{proof}
From the definition of $\infl^2_i(f)$, we get that $\infl^2(f) = \sum_{S\subseteq[n]}|S|\hat{f}(S)^2$. Hence
$$\sum_{S\subseteq[n],\ |S| > d} \hat{f}(S)^2 \leq \fr{d} \infl^2(f)\  .$$
\end{proof}

This gives a simple proof that submodular and XOS functions are $\epsilon$-approximated in $\ell_2$ by polynomials of degree $2/\eps^2$ (which was proved for submodular functions in \cite{CheraghchiKKL:12}).
Next, we show a stronger statement, that these functions are $\eps$-approximated by $2^{O(1/\eps^2)}$-juntas of Fourier degree $O(1/\eps^2)$.

\subsection{Friedgut's Theorem for Real-Valued Functions}
\label{sec:friedgut-rv}
As we have shown in Lemma \ref{lem:submod}, self-bounding functions have low total $\ell_1$-influence. A celebrated result of Friedgut \cite{Friedgut:98} shows that any Boolean function on $\zo^n$ of low total influence is close to a function that depends on few variables. It is therefore natural to try and apply Friedgut's result to our setting. A commonly considered generalization of Boolean influences to real-valued functions uses $\ell_2^2$-influences which can be easily expressed using Fourier coefficients (\eg \cite{DinurFKO:06}). However, a Friedgut-style result is not true for real-valued functions when $\ell_2^2$-influences are used, as observed by O'Donnell and Servedio \cite{ODonnellServedio:07} (see also Sec.~\ref{app:l2-example}). This issue also arises in the problem of learning real-valued monotone decision trees \cite{ODonnellServedio:07}. They overcome the problem by first discretizing the function and proving that Friedgut's theorem can be extended to the discrete case (as long as the discretization step is not too small). The problem with using this approach for submodular functions is that it does not preserve submodularity and can increase total influence of the resulting function to $\Omega(\sqrt{n})$ with discretization parameters necessary for the approach to work (consider for example a linear function $\fr{n}\sum_i x_i$).

Here we instead prove a generalization of Friedgut's theorem to all real-valued functions. We show that Friedgut's theorem holds for real-valued functions if the total $\ell_\kappa^\kappa$-influence  (for some constant $\kappa \in [1, 2)$) is small in addition to total $\ell_2^2$-influence. Self-bounding functions have low total $\ell_1$-influence and hence for our purposes $\kappa =1$ would suffice. We prove the slightly more general version as it could be useful elsewhere (and the proof is essentially the same).
\begin{theorem}
\label{th:rvjunta}
Let $f:\zo^n\rightarrow \RR$ be any function, $\eps \in (0,1)$ and $\kappa \in (1,2)$. For $d$ such that $\sum_{|S|>d} \hat{f}(S)^2 \leq \eps/2$, let $$I = \{i\in[n] \cond \infl^{\kappa}_i(f) \geq \alpha\} \mbox{ for}$$ $$\alpha = \left((\kappa-1)^{d-1} \cdot \eps/ (2 \cdot\infl^{\kappa}(f)) \right)^{\kappa/(2-\kappa)}\ .$$
Then for the set $\I_d =\{S \subseteq I \cond |S| \leq d\}$ we have $\sum_{S\not\in \I_d} \hat{f}(S)^2 \leq \eps$.
\end{theorem}
To obtain Theorem \ref{thm:Friedgut-junta-intro} from this statement we use it with $\eps^2$ error and let $d= 2 \cdot  \infl^2(f)/\epsilon^2$ which, by Lemma \ref{lem:low-degree-conc}, gives the desired bound on $\sum_{|S|>d} \hat{f}(S)^2$. Note that $g = \sum_{S \in \I_d} \hat{f}(S) \chi_S$ is a function of Fourier degree $d$ that depends only on variables in $I$. Further, $\|f-g\|_2^2 \leq \eps^2$ and the set $I$ has size at most
\equ{\label{eq:bound-i} |I| \leq \infl^{\kappa}(f)/\alpha = 2^{O(\infl^2(f)/\epsilon^2)} \cdot \eps^{2\kappa/(2-\kappa)} \cdot (\infl^{\kappa}(f))^{2/(2-\kappa)}.}
Also note that Theorem \ref{th:rvjunta} does not allow us to directly bound $|I|$ in terms of $\infl^1(f)$ since it does not apply to $\kappa =1$. However for every $\kappa \in [1,2]$, $\infl^{\kappa}(f) \leq \infl^1(f) + \infl^2(f)$ and therefore we can also bound $|I|$ using equation (\ref{eq:bound-i}) for $\kappa=4/3$ and then substituting $\infl^{4/3}(f) \leq \infl^1(f) + \infl^2(f)$. This gives the proof of Theorem \ref{thm:Friedgut-junta-intro} (first part). The second part of Theorem \ref{thm:Friedgut-junta-intro} now follows from Lemma \ref{lem:submod}.

Our proof of Theorem \ref{th:rvjunta} is a simple modification of the proof of Friedgut's theorem from  \cite{DinurFriedgut:05course}. We will need the notion of a noise operator.

\begin{definition}[The noise operator]
For $\alpha \in [0,1], x \in \zon$, we define a distribution $N_\alpha(x)$ over $y \in \zon$ by letting $y_i = x_i$ with probability $1-\alpha$ and $y_i = 1-x_i$ with probability $\alpha$, independently for each $i$. For $\rho \in [-1,1]$, the noise operator $T_\rho$ acts on functions $f:\zon \rightarrow \RR$, and is defined by $$(T_\rho f)(x) = \E_{y \sim N_{1/2-\rho/2}(x)}[f(y)].$$ In the Fourier basis the noise operator satisfies: $\widehat{(T_\rho f)}(S) = \rho^{|S|}\hat{f}(S)$, for every $S \subseteq [n]$.
\end{definition}

Following Friedgut's proof, we will require a bound on $\|T_\rho f\|_2$ in terms of $\|f\|_\kappa$. This lemma is a special case of the Hypercontractive inequality of Bonami and Beckner \cite{Bon70,Bec75}.
\begin{lemma}
\label{lem:hypercont}
For any $f:\zo^n\rightarrow \RR$, and any $\kappa \in [1,2]$,  $\|T_{\sqrt{\kappa-1}} f\|_2 \leq \| f \|_\kappa$.
\end{lemma}

The proof of Theorem \ref{th:rvjunta} relies on two lemmas. The first one is Lemma~\ref{lem:low-degree-conc}, stated above.
The second and key lemma is the following bound on the sum of squares of all low-degree Fourier coefficients that include a variable of low influence.
\begin{lemma}
\label{lem:low-influence-bound-general}
Let $f:\zo^n\rightarrow \RR$, $\kappa \in (1,2)$, $\alpha >0$ and $d$ be an integer $\geq 1$. Let $I = \{i\in[n] \cond \infl^{\kappa}_i(f) \geq \alpha\}$. Then $$\sum_{S\not\subseteq I, |S|\leq d}\hat{f}(S)^2 \leq (\kappa-1)^{1-d} \cdot \alpha^{2/\kappa-1} \cdot \infl^{\kappa}(f)\ .$$
\end{lemma}
\begin{proof}
We first observe that by the properties of the Fourier transform of $\partial_i f$ (see Sec.~\ref{sec:lowsens-prelims}) and the noise operator $T_\rho$, we have
\equ{\label{eq:partial-noise}
\left\|T_\rho \frac{\partial_i f}{2}\right\|_2^2 = \sum_{S\subseteq[n], S \ni i} (\rho^2)^{|S|-1} \hat{f}(S)^2.
}
Next we bound the sum in terms of norms of $T_{\sqrt{\kappa-1}}$ applied to $\partial_i f$'s.
\alequn{\sum_{S\not\subseteq I, |S|\leq d}\hat{f}(S)^2 &\leq \sum_{S\subseteq[n], |S|\leq d}|S\cap \bar{I}|\hat{f}(S)^2 \leq (\kappa-1)^{1-d} \sum_{S\subseteq[n], |S|\leq d}|S\cap \bar{I}| (\kappa-1)^{|S|-1}\hat{f}(S)^2 \\ & =
(\kappa-1)^{1-d}\sum_{i\in \bar{I}} \sum_{S\subseteq[n], S \ni i} (\kappa-1)^{|S|-1} \hat{f}(S)^2 = (\kappa-1)^{1-d}\sum_{i\in \bar{I}}\left\|T_{\sqrt{\kappa-1}} \frac{\partial_i f}{2}\right\|_2^2,}
where the last equality follows from eq.~(\ref{eq:partial-noise}).
Now we can apply Lemma \ref{lem:hypercont} to obtain:
\alequn{\sum_{i\in \bar{I}}\left\|T_{\sqrt{\kappa-1}} \frac{\partial_i f}{2}\right\|_2^2 &\leq \sum_{i\in \bar{I}}\left\|\frac{\partial_i f}{2}\right\|_{\kappa}^2 = \sum_{i\in \bar{I}}\E\left[\left|\frac{\partial_i f}{2}\right|^{\kappa}\right]^{\frac{1}{\kappa}\cdot 2}\\& = \sum_{i\in \bar{I}} \left(\infl_i^{\kappa}(f)\right)^{2/\kappa}  \\& \leq  \max_{i\in \bar{I}} \left(\infl_i^{\kappa}(f)\right)^{2/\kappa-1} \cdot \sum_{i\in \bar{I}} \infl_i^{\kappa}(f) \\& \leq \cdot \alpha^{2/\kappa-1} \cdot \infl^{\kappa}(f).}
\end{proof}

We now proceed to obtain Theorem \ref{th:rvjunta} by combining Lemmas \ref{lem:low-degree-conc} and \ref{lem:low-influence-bound-general}.

\begin{proof}[Proof of Thm.~\ref{th:rvjunta}]

Observe that
$$\sum_{S \not\in \I_d} \hat{f}(S)^2 = \sum_{S\subseteq[n], |S|>d}\hat{f}(S)^2 + \sum_{S\not\subseteq I, |S|\leq d}\hat{f}(S)^2\ .$$
For our choice of $d$, 
$\sum_{S\subseteq[n], |S|>d}\hat{f}(S)^2 \leq \eps/2$.

Now, by Lemma \ref{lem:low-influence-bound-general} the second part can be bounded by
$$\sum_{S\not\subseteq I, |S|\leq d}\hat{f}(S)^2 \leq (\kappa-1)^{1-d} \cdot \alpha^{2/\kappa-1} \cdot \infl^{\kappa}(f) =
(\kappa-1)^{1-d} \cdot  \left((\kappa-1)^{d-1} \cdot \eps/(2\infl^{\kappa}(f)) \right) \cdot \infl^{\kappa}(f) = \eps/2\ .$$
\end{proof}

We now give a slightly simpler version of Thm.~\ref{th:rvjunta} for functions that have low total $\ell_1$-influence, such as self-bounding functions.
\begin{corollary}
\label{cor:lowsens-junta-l1-spectral-bound}
Let $f:\zo^n \rightarrow [0,1]$ be any function and $\eps>0$. For $d = 2\cdot \infl^1(f)/\epsilon^2$ and $\alpha = 2^{-4d}$ let
$$I = \{ i\in[n] \cond \infl^1_i(f) \geq \alpha\}.$$ There exists a function $p$ of Fourier degree $d$ over variables in $I$, such that $\|f - p\|_2 \leq \epsilon$ and $\|\hat{p}\|_1 \leq 2^{O(\infl^1(f)^2/\epsilon^4)}$.
\end{corollary}
\begin{proof}
We first note that, for every $i$, $\frac{\partial_i f}{2}$ has range in $[-1,1]$ and therefore for every $\kappa \geq 1$,
$$\infl_i^\kappa(f) = \E\left[\left|\frac{\partial_i f}{2}\right|^\kappa\right] \leq \E\left[\left|\frac{\partial_i f}{2}\right|\right] = \infl_i^1(f). $$
In particular, $\infl^2(f) \leq \infl^1(f)$ and $\infl^{4/3}(f) \leq \infl^1(f)$. We can now apply  Thm.~\ref{th:rvjunta} with $\kappa=4/3$ to obtain that for $d = 2\cdot\infl^1(f)/\epsilon^2 \geq 2\cdot\infl^2(f)/\epsilon^2$, $\alpha = 2^{-4d} \leq \left(3^{-d+1} \cdot \eps^2/ (2\cdot\infl^{4/3}(f)) \right)^2$ and
$$I' = \{ i\in[n] \cond \infl^{4/3}_i(f) \geq \alpha\} $$ we have that
$$\sum_{S \not\subseteq I'\mbox{ or } |S| > d} (\hat{f}(S))^2 \leq \eps^2 .$$
Let $p = \sum_{S \subseteq I',\ |S|\leq d} \hat{f}(S)\chi_S$. Then
$\|f - p\|_2^2 \leq \epsilon^2$. 
Now we observe that $\infl_i^{4/3}(f) \leq \infl_i^1(f)$ implies that $I' \subseteq I$ and therefore $p$ is a function of Fourier degree $d$ over variables in $I$. To bound $\|\hat{p}\|_1$ we observe that $|I| \leq \infl^1(f)/\alpha$ and therefore the total number of non-zero Fourier coefficients of $p$ is at most $$\sum_{j\leq d} {|I| \choose j} \leq |I|^d = (2^{4d} \cdot \infl^1(f))^{d} =  2^{O(\infl^1(f)^2/\epsilon^4)}\ .$$ To get the desired bound on  $\|\hat{p}\|_1$ it now suffices to note that $f$ has range  $[-1,1]$ and therefore for every $S \subseteq [n]$, $|\hat{p}(S)| \leq |\hat{f}(S)| \leq 1$.
\end{proof}

\eat{
For the special case of XOS functions we can further distill the following structural corollary.
\begin{corollary}
\label{cor:XOS-junta-bound}
Let $f:\zo^n \rightarrow [0,1]$ be an XOS function and $\eps>0$. There exists a $2^{8/\eps^2}$-junta $p$ of Fourier degree $2/\eps^2$, such that $\|f - p\|_2 \leq \epsilon$.
\end{corollary}
}

\section{Lower Bound Examples}
\label{sec:lower-bounds}
Here we show three simple examples: The first one shows that Theorem~\ref{thm:submod-junta} is almost optimal, in the sense that the dependence on $\epsilon$ cannot be better than $1/\epsilon^2$. The second example shows that  Corollary~\ref{cor:lowsens-junta-l1-spectral-bound} is essentially optimal even for XOS functions. Finally, the third example shows that Theorem~\ref{th:rvjunta} requires the use of $\ell_\kappa^\kappa$-influences for $\kappa < 2$ rather than just $\ell_2^2$-influences.

\subsection{Lower Bound On Junta Size For Linear Functions}

We prove that even for linear functions, an $\epsilon$-approximation (even in $\ell_1$-norm) requires at least $1/\epsilon^2$ variables.

\begin{lemma}
\label{lem:linear-example}
Consider a linear function $f:\{0,1\}^n \rightarrow [0,1]$, $$f(x) = \frac{1}{a} \sum_{i \in A} x_i$$ where $|A| = a$. Then every function $g:\{0,1\}^n \rightarrow \RR$ that depends on less than $\frac{a}{2}$ variables has $\|f-g\|_1 = \Omega(\sqrt{1/a})$.
\end{lemma}

\begin{proof}
Suppose that $g$ depends only on a subset of variables $B$. Denote by $f_{x_B}$ the restriction of $f$ to $\{0,1\}^{\bar{B}}$ after the coordinates on $B$ have been fixed to $x_B$. Note that $f_{x_B}$ is still a linear function. Hence, the closest function to $f$ depending only on $x_B$ (whether in $\ell_1$ or $\ell_2$) is $g(x) = \E[f_{x_B}] = \frac{1}{a} \sum_{i \in B} x_i + \frac{1}{2a} |A \setminus B|$.

Let us compute the distance between $f$ and $g$: After fixing the coordinates on $B$, $f_{x_B}$ is a linear function of variance
$$ \Var[f_{x_B}] = \sum_{i \in A \setminus B} \frac{1}{a^2} \Var[x_i] = |A \setminus B| \cdot \frac{1}{4 a^2}.$$
This means that with constant probability, $f_{x_B}$ deviates from its expectation by at least $\sqrt{\Var[f_{x_B}]} = \frac{1}{2a} \sqrt{|A \setminus B|}$. Consequently, $|f(x) - g(x)| > \frac{1}{2a} \sqrt{|A \setminus B|}$ with constant probability and $\|f-g\|_1 = \Omega(\frac{1}{2a} \sqrt{|A \setminus B|})$. If $|A \setminus B| \geq \frac{a}{2}$, then we obtain $\|f-g\|_1 = \Omega(\sqrt{1/a})$.
\end{proof}

\subsection{Lower Bound On Junta Size For XOS Functions}

Here we prove that Theorem \ref{thm:Friedgut-junta-intro} is close-to-tight and, in particular, Theorem~\ref{thm:submod-junta} cannot be extended to XOS functions. In fact, we show that $2^{\Omega({1}/{\epsilon})}$ variables are necessary for an $\epsilon$-approximation to an XOS function. Our lower bound is based on the Tribes DNF function studied by Ben-Or and Linial \cite{Ben-OrLinial:85} with AND replaced by a linear function. The Tribes DNF was also used by Friedgut to prove tightness of his theorem for Boolean functions \cite{Friedgut:98}.

\begin{theorem}
\label{thm:XOS-example}
Suppose that $n = a b$ where $b = 2^a$ and consider an XOS function
$$f(x) = \frac{1}{a} \max_{1 \leq j \leq b} \sum_{i \in A_j} x_i$$
 where $(A_1,\ldots,A_b)$ is a partition of $[n]$ into sets of size $|A_j| = a$. Then every function $g:\{0,1\}^n \rightarrow \RR$ that depends on fewer than $2^{a-1}$ variables has $\|f-g\|_1 = \Omega(1/a)$.
\end{theorem}
\begin{proof}
Suppose that $g$ depends on fewer than $2^{a-1}$ variables. This means that there are fewer than $2^{a-1}$ parts where $g$ depends on any variable. Let us denote the parts where $g$ does not depend on any variable by $\cal D$; we have $|{\cal D}| > 2^{a-1}$.

We observe the following: For each part, $\Pr[\sum_{i \in A_j} x_i = a] = 2^{-a}$ (all $a$ variables should be equal to $1$). Therefore, with probability at least $(1 - 2^{-a})^{2^{a-1}} \simeq e^{-1/2}$ we have $\sum_{i \in A_j} x_i < a$ for all $j \notin {\cal D}$. Let us condition on some values of $\{x_i: i \in \bigcup_{j \notin {\cal D}} A_j \}$ such that this is the case. In this event, $f(x) = 1$ iff we have $\sum_{i \in A_j} x_i = a$ for at least one of the parts $j \in {\cal D}$. This happens with constant probability (since $2^{a-1} < |{\cal D}| \leq 2^a$), bounded away from both $0$ and $1$. Hence, $f(x)$ is either $1$ or at most $1-1/a$, both with constant nonzero probabilities.

On the other hand, our function $g$ does not depend on the variables in $\bigcup_{j \in {\cal D}} A_j$ at all. Therefore, given the variables $\{x_i: i \in \bigcup_{j \notin {\cal D}} A_j \}$, $g(x)$ has a fixed value, and with constant probability it differs from $f(x)$ by at least $\frac{1}{2a}$. Overall, this happens with constant probability, and hence $\|f-g\|_1 = \Omega(1/a)$.
\end{proof}

\subsection{Lower Bound For Total $\ell_2^2$-influence}
\label{app:l2-example}

Here we show that a generalization of Friedgut's theorem to real-valued functions cannot use total $\ell_2^2$-influence only. A similar example also appears in \cite{ODonnellServedio:07}.

\begin{lemma}
There is an absolute constant $\alpha>0$ and a function $f:\on^n \rightarrow [-1,1]$ for any $n$, such that $\infl^2(f) \leq 1$, and for any function $g$ that depends only on $n/2$ variables, $\|f-g\|_1 \geq \alpha$.
\end{lemma}

\begin{proof}
Let
\begin{itemize}
\item $f(x) = \frac{1}{\sqrt{n}} \sum_{i=1}^{n} x_i$ for $|\sum_{i=1}^{n} x_i| \leq \sqrt{n}$,
\item $f(x) = 1$ for $\sum_{i=1}^{n} x_i > \sqrt{n}$, and
\item $f(x) = -1$ for for $\sum_{i=1}^{n} x_i < -\sqrt{n}$.
\end{itemize}

The total $\ell_2^2$-influence is easy to estimate:
$$ \infl^2(f) = \sum_{i=1}^{n} \E[(\frac12 \partial_i f(x))^2]  \leq n \cdot \frac{1}{n} = 1.$$

Now assume that $g:\on^n \rightarrow \RR$ depends only on a subset of coordinates $J$, $|J| = n/2$.
Condition on any choice of values for $x_J$ such that $|\sum_{i \in J} x_i| < \sqrt{n}$. (This happens with constant probability for random $x_J$.) The remaining $n/2$ variables satisfy with constant probability $(\sgn \ g(x_J)) (\sum_{i \notin J} x_i) < -2\sqrt{n}$ (recall that $g$ depends only on $x_J$). This implies that $(\sgn \ g(x)) \cdot \sum_{i=1}^{n} x_i < -\sqrt{n}$; i.e., $g(x)$ and $f(x)$ have opposite signs and moreover $|\sum_{i=1}^{n} x_i| > \sqrt{n}$, so $|f(x)| = 1$. Thus with constant probability, $|f(x) - g(x)| \geq 1$.
\end{proof}

\section{Applications to PAC Learning}
\label{sec:applications-learning}
We now show that our approximation of submodular and low-influence functions by juntas can be used to give faster PAC and PMAC learning algorithms for these classes of functions.

\subsection{Preliminaries: Models of Learning}

We consider two models of learning based on the PAC model \cite{Valiant:84} which assumes that the learner has access to random examples of an unknown function from a known class of functions. In the first model we measure the performance of the learner by $\ell_1$-error between the target and the hypothesis, which generalizes the notion of disagreement error used for learning Boolean functions (\eg \cite{Haussler:92}).
\begin{definition}[PAC learning with $\ell_1$-error]
Let $\F$ be a class of real-valued functions on $\zo^n$ and let $D$ be a distribution on $\zo^n$. An algorithm $\A$ PAC-learns $\F$ on $D$, if given $\epsilon > 0$, for every target function $f \in \F$, given access to random  independent samples from $D$ labeled by $f$, with probability at least $2/3$,  $\A$ returns a hypothesis $h$ such that $\E_{x \sim D} [|f (x) - h(x) | ] \leq  \epsilon. $ $\A$ is said to be \emph{proper} if $h \in \F$. $\A$ is said to be \emph{efficient} if $h$ can be evaluated in polynomial time on any input and the running time of $\A$ is polynomial in $n$ and $1/\epsilon$.
\end{definition}
In some cases we bound the $\ell_2$-error of the hypothesis which also upper-bounds its $\ell_1$-error. While in general Valiant's model does not make assumptions on the distribution $D$, here we only consider the {\em distribution-specific} version of the model in which the distribution is fixed and is uniform over $\zo^n$.

The second model that we consider is the PMAC model introduced by Balcan and Harvey \cite{BalcanHarvey:12full}
which requires a multiplicative-factor approximation of the target function.
A PMAC learner with approximation factor $\alpha$ and error $\eps$ is an algorithm which outputs a hypothesis $h$ that satisfies $\Pr_{x \sim D} [f (x) \leq h(x) \leq \alpha f(x) ] \geq  1-\epsilon $. We say that $h$ multiplicatively $(\alpha,\eps)$-approximates $f$ over $D$ in this case.\footnote{The definition in \cite{BalcanHarvey:12full} uses the condition $h(x) \leq f(x) \leq \alpha h(x)$ which is equivalent up to scaling the hypothesis by $\alpha$.}

\subsection{Finding Influential Variables}
In order to exploit the fact that a submodular function can be approximated by a junta we need to find the variables of the junta. Unfortunately, the criterion for including variables given in Algorithm \ref{alg:junta} cannot be (efficiently) evaluated using random examples alone. Instead we give a general way to find a larger approximating junta whenever an approximating junta exists. For a real-valued $f$ over $\zo^n$ and $\eps \in [0,1]$ let $s_f(\eps)$ denote the smallest $s$ such that there exists an $s$-junta $g$ for which $\|f-g\|_2 \leq \eps$. For a set of indices $I \subseteq [n]$ we say that a function is an $I$-junta if it depends only on variables in $I$.
\begin{theorem}
\label{thm:find-junta-examples}
Let $f:\zo^n \rightarrow [0,1]$ be a submodular function. There exists an algorithm, that given any $\eps > 0$ and access to random and uniform examples of $f$, with probability at least $5/6$, finds a set of variables $I$ of size at most $32 \cdot (s_f(\eps/2))^2/\eps$ such that there exists a submodular $J$-junta $h$ for $J \subseteq I$ of size $s_f(\eps/2)$ satisfying $\|f -h\|_1 \leq \eps$. The algorithm runs in time $O(n^2 \log (n) \cdot (s_f(\eps/2))^4/\eps^2)$ and uses $O(\log (n) \cdot (s_f(\eps/2))^4/\eps^2)$ examples.
\end{theorem}
Our algorithm selects all variables that have a large degree-1 or 2 Fourier coefficient. This is the same algorithm as the one used in \cite{FeldmanKV:13} (with different values of thresholds). However the analysis in \cite{FeldmanKV:13} relies crucially on the spectral $\ell_1$-norm of an $\eps$-approximating function $g$ and gives a junta of size $\poly(\|\hat{g}\|_1)$. As can be seen from the lower bound in \cite{FeldmanKV:13}, the spectral $\ell_1$-norm of any function that $\eps$-approximates certain submodular functions must be exponential in $1/\eps$ and therefore this argument is not useful for our purposes. Instead we give a new and more general argument that relies on the fact that total $\ell_1$-influence of submodular functions is upper-bounded by a constant (Lemma \ref{lem:submod}).

For a function $f$ and a set of indices $I$, we define the {\em projection} of $f$ to $I$ to be the function over $\zo^n$ whose value depends only on the variables in $I$ and its value at $x_I$ is the expectation of $f$ over all the possible values of variables outside of $I$, namely $f_I(x) = \E_{y\sim \U}[f(x_I,y_{\bar{I}})]$.
We start by establishing several simple properties of projections and influences.
\begin{lemma}
\label{lem:properties}
Let $f:\zo^n \rightarrow \RR$ be any function, $i \in [n]$ and $I \subseteq [n]$. Then
\begin{enumerate}
\item for every $I$-junta $h$, $\|f-h\|_2 \geq \|f-f_I\|_2$; \label{itm:closest}
\item If $i \in I$ then $(\partial_i f)_I = \partial_i f_I$; \label{itm:partial-swap}
\item $\infl^1_{i}(f) \leq \|f\|_1$; \label{itm:infl-level}
\item $\infl^1_{i}(f_I) \leq \infl^1_{i}(f)$; \label{itm:infl-monot}
\item $|\hat{f}(\{i\})| \leq \infl^1_{i}(f)$; \label{itm:infl-from-chow}
\item \cite{FeldmanKV:13} for all $j \neq i$, $|\hat{f}(\{i,j\})| = \infl^1_{i}(\partial_j f)/2$; \label{itm:second-is-infl}
\item $\|f - f_I\|_1 \leq \sum_{j \not\in I} \infl^1_i(f)$; for all $J \subseteq [n]$, $\|f_J - f_{I\cap J}\|_1 \leq \sum_{j \in J\setminus I} \infl^1_j(f_J)$. \label{itm:project-error}
\end{enumerate}
\end{lemma}
\begin{proof}
\begin{enumerate}
\item As is well-known, for any set of $m$ real values $a_1,\ldots,a_m$, the value of $\sum_i (b-a_i)^2$ is minimized when $b = \frac{1}{m}\sum a_i$. Therefore $f_I$ is the $I$-junta closest (in $\ell_2$-norm) to $f$.
\item For $b\in \zo$ let $f_b$ be defined as $f_{i \leftarrow b}(x) = f(x_{i \leftarrow b})$. First, observe that if $i \in I$ then we can exchange the restriction and projection operators on $f$, that is, for every $x$, $f_{i \leftarrow b,I}(x) = f_{I,i \leftarrow b}(x)$. Now
$$(\partial_i f)_I = \left(f_{i \leftarrow 1}-f_{i \leftarrow 0}\right)_I = f_{i \leftarrow 1,I} - f_{i \leftarrow 0,I} = f_{I,i \leftarrow 1} - f_{I,i \leftarrow 0} = \partial_i f_I\ .$$
\item
    $$\infl^1_{i}(f) = \E\left[\left|\frac{f(x_{i \leftarrow 1}) - f(x_{i \leftarrow 0})}{2}\right|\right] \leq \E\left[\frac{|f(x_{i \leftarrow 1})| + |f(x_{i \leftarrow 0})|}{2}\right] = \E[|f(x)|] = \|f\|_1 .$$
\item Convexity of $|\cdot |$ implies that for every function $g:\zo^n \rightarrow \RR$, $\|g_I\|_1 \leq \|g\|_1$. Together with property (\ref{itm:partial-swap}) this implies that
$$\infl^1_i(f_I) = \E[|\partial_i f_I|]/2 =  \E[|(\partial_i f)_I|]/2 \leq \E[|\partial_i f|]/2 = \infl^1_{i}(f)\ .$$
\item $$|\hat{f}(\{i\})| = |\E[\partial_i f]|/2 \leq \E[|\partial_i f|]/2 = \infl^1_{i}(f)\ .$$
\item $$|\hat{f}(\{i,j\})| =^{*} \fr{4} |\E[\partial_{i,j} f]| =^{**} \fr{4} \E_\U[|\partial_{i,j} f|] = \fr{2}\infl^1_{i}(\partial_j f) .$$
Here, $(*)$ follows from the basic properties of the Fourier spectrum of partial derivatives (see Sec.~\ref{sec:lowsens-prelims}) and $(**)$ is implied by second partial derivatives of a submodular function being always non-positive (see Sec.~\ref{sec:prelims}).
\item First,
\alequn{\|f- f_{[n]\setminus\{j\}}\|_1 &= \E\left[\left|\frac{f(x_{j \leftarrow 0}) + f(x_{j \leftarrow 1})}{2} - f(x)\right|\right] \leq \E\left[\frac{\left|f(x_{j \leftarrow 0}) -f(x)\right|}{2}\right] + \E\left[\frac{\left|f(x_{j \leftarrow 1}) -f(x)\right|}{2}\right] \\ &= \E[|(f(x_{j \leftarrow 1}) - f(x_{j \leftarrow 0})|/2] = \E[|\partial_j(f)|]/2 = \infl^1_j(f)\ .}
Together with property (\ref{itm:infl-monot}), this implies that for any $j \not\in I$, $\|f_{I\cup \{j\}} - f_I \|_1 \leq \infl^1_j(f_{I\cup \{j\}}) \leq \infl^1_j(f)$.
By applying this iteratively to all $j \not\in I$ and using the triangle inequality we obtain that $$\|f - f_I\|_1 \leq \sum_{j \not\in I} \infl^1_j(f)\ .$$ To obtain the second part we apply the first part to $f_J$ and obtain
$$\|f_J - f_{I\cap J}\|_1 \leq \sum_{j \not\in I} \infl^1_j(f_J)\ .$$ Observe that for all $j \not\in J$, $\infl^1_j(f_J) = 0$ and hence
$$\sum_{j \not\in I} \infl^1_j(f_J) = \sum_{j \in J\setminus I} \infl^1_j(f_J) \ .$$
\end{enumerate}
\end{proof}

We now prove that throwing away variables with small degree-1 or 2 Fourier coefficients does not affect a projection of $f$ to a small set of variables $J$ significantly.
\begin{lemma}
\label{lem:junta-error-bound}
Let $f:\zo^n \rightarrow \RR$ be a real-valued function and let $J \subseteq [n]$. Let
$$I' = \left\{ i \lcond |\hat{f}(\{i\})| \geq \frac{\eps}{2\cdot|J|} \right.\right\} \bigcup \left\{ i \lcond \exists j, |\hat{f}(\{i,j\})| \geq \frac{\eps}{2 \cdot |J|^2} \right.\right\}\ $$ and let $I \supseteq I'$. Then $\|f_J - f_{I\cap J}\|_1 \leq \eps$.
\end{lemma}
\begin{proof}
By Lem.~\ref{lem:properties}(\ref{itm:project-error}) we obtain that
\equ{\|f_J - f_{I\cap J}\|_1 \leq \sum_{i \in J\setminus I} \infl^1_i(f_J) = \fr{2}\sum_{i \in J\setminus I} \|\partial_i f_J\|_1. \label{eq:bound-total-error}} We now apply Lem.~\ref{lem:properties}(\ref{itm:project-error}) to $\partial_i f$ and the empty set projection:
\equ{\|(\partial_i f)_J - (\partial_i f)_\emptyset\|_1 \leq \sum_{j \in J\setminus\{i\}} \infl^1_j((\partial_i f)_J)\ . \label{eq:bound-partial}}
By Lem.~\ref{lem:properties}(\ref{itm:infl-monot},\ref{itm:second-is-infl}), $\infl^1_j((\partial_i f)_J) \leq \infl^1_j(\partial_i f) = 2|\hat{f}(\{i,j\})|$. For $i \not\in I$, $|\hat{f}(\{i,j\})| \leq \eps/(2|J|^2)$. By substituting this into equation (\ref{eq:bound-partial}) we get that
$$\|(\partial_i f)_J - (\partial_i f)_\emptyset\|_1 \leq \sum_{j \in J\setminus\{i\}} 2 \cdot\frac{\eps}{2 \cdot |J|^2} \leq \frac{\eps}{|J|} .$$
Now we note that $(\partial_i f)_\emptyset \equiv \E[\partial_i f] = -2\hat{f}(\{i\})$ and (by Lem.~\ref{lem:properties}(\ref{itm:partial-swap})) $(\partial_i f)_J = \partial_i f_J$. This implies that for $i \not\in I$,
$$ \|\partial_i f_J\|_1 \leq \|\partial_i f_J - (\partial_i f)_\emptyset\|_1 + \|(\partial_i f)_\emptyset\|_1 \leq
\frac{\eps}{|J|} +  2|\hat{f}(\{i\})| \leq \frac{2\eps}{|J|}\ .$$
Substituting this into equation (\ref{eq:bound-total-error}) we obtain that
$$\|f_J - f_{I\cap J}\|_1 \leq \fr{2}\sum_{i \in J\setminus I} \|\partial_i f_J\|_1 \leq  \fr{2}\sum_{i \in J\setminus I}  \frac{2\eps}{|J|} \leq \eps\ . $$
\end{proof}

We next bound the number of variables that have large degree-1 or degree-2 Fourier coefficient (a weaker bound is also implied by Parseval's identity).
\begin{lemma}
\label{lem:important-var-bound}
Let $f:\zo^n \rightarrow [0,1]$ be a submodular function and $\alpha,\beta > 0$. Let
$$I = \left\{ i \lcond |\hat{f}(\{i\})| \geq \alpha \right.\right\} \bigcup \left\{ i \lcond \exists j, |\hat{f}(\{i,j\})| \geq \beta \right.\right\}\ .$$ Then $|I| \leq \frac{2}{\min\{\alpha,\beta\}}$.
\end{lemma}
\begin{proof}
If $i\in I$ then either $|\hat{f}(\{i\})| \geq \alpha$ or $|\hat{f}(\{i,j\})| \geq \beta$ for some $j\neq i$. In the former case, by Lem.~\ref{lem:properties}(\ref{itm:infl-from-chow}), $\infl^1_i(f) \geq |\hat{f}(\{i\})| \geq \alpha$ and in the latter case, by Lem.~\ref{lem:properties}(\ref{itm:infl-level},\ref{itm:second-is-infl})
$$ \infl^1_i(f) = \fr{2}\|\partial_i f\|_1 \geq \fr{2}\infl^1_{j}(\partial_i f)= |\hat{f}(\{i,j\})| \geq \beta\ .$$
This implies that for all $i \in I$, $\infl^1_i(f) \geq \min\{\alpha,\beta\}$. By Lemma \ref{lem:submod}, $\infl^1(f) = \sum_{i \in [n]} \infl^1_i(f) \leq 2$. This gives the claimed bound on $|I|$.
\end{proof}

We are now ready to complete the proof of Theorem \ref{thm:find-junta-examples}.
\begin{proof}[Proof of Theorem \ref{thm:find-junta-examples}]
Let $J \subseteq [n]$ be a set of indices of size $s_f(\eps/2)$ such that there exists a $J$-junta $g$ for which $\|f-g\|_2 \leq \eps/2$. By Lem.~\ref{lem:properties}(\ref{itm:closest}), this implies that $\|f-f_J\|_1 \leq \|f-f_J\|_2 \leq \|f-g\|_2 \leq \eps/2$.
Let $$I' = \left\{ i \lcond |\hat{f}(\{i\})| \geq \frac{\eps}{4\cdot s_f(\eps/2)} \right.\right\} \bigcup \left\{ i \lcond \exists j, |\hat{f}(\{i,j\})| \geq \frac{\eps}{8 \cdot (s_f(\eps/2))^2} \right.\right\}\ .$$
By Lemma \ref{lem:junta-error-bound}, for any $I \supseteq I'$, $\|f_J - f_{I\cap J}\|_1 \leq \eps/2$. In particular, it is easy to see that  
$f_{I\cap J}$ is a submodular $(I\cap J)$-junta. Clearly, $J \cap I \subseteq I$ and $|J \cap I| \leq s_f(\eps/2)$. By the triangle inequality, $\|f - f_{I\cap J}\|_1 \leq \eps$.

All we need now is to find a small set of indices $I \supseteq I'$. We simply estimate degree-1 and 2 Fourier coefficients of $f$ to accuracy $\eps/(32 \cdot (s_f(\eps/2))^2)$ with confidence at least $5/6$ using random examples. Let $\tilde{f}(S)$ for $S\subseteq [n]$ of size 1 or 2 denote the obtained estimates. We define
$$I = \left\{ i \lcond |\tilde{f}(\{i\})| \geq \frac{3\eps}{16 \cdot s_f(\eps/2)} \right.\right\} \bigcup \left\{ i \lcond \exists j, |\tilde{f}(\{i,j\})| \geq \frac{3\eps}{32 \cdot (s_f(\eps/2))^2} \right.\right\}\ .$$
If estimates are within the desired accuracy, then clearly, $I \supseteq I'$. At the same time $I \subseteq I''$, where
$$I'' = \left\{ i \lcond |\hat{f}(\{i\})| \geq \frac{\eps}{8\cdot s_f(\eps/2)} \right.\right\} \bigcup \left\{ i \lcond \exists j, |\hat{f}(\{i,j\})| \geq \frac{\eps}{16 \cdot (s_f(\eps/2))^2} \right.\right\}\ .$$ By Lem.~\ref{lem:important-var-bound}, $|I''| \leq 32 \cdot (s_f(\eps/2))^2/\eps$.

Finally, to bound the running time we observe that, by the standard application of Chernoff bound with the union bound, $O(\log (n) \cdot  (s_f(\eps/2))^4/\eps^2)$ random examples are sufficient to obtain the desired estimates with confidence of $5/6$. The estimation of the coefficients can be done in $O(n^2 \log(n) \cdot (s_f(\eps/2))^4/\eps^2)$ time.
\end{proof}

Our main structural result together with Theorem \ref{thm:find-junta-examples} imply that, given random examples of a submodular function $f$, one can find $\tilde{O}(1/\eps^{5})$ variables such that there exists a submodular $\tilde{O}(1/\eps^2)$-junta over those variables $\eps$-close to $f$.
\begin{corollary}
\label{cor:find-junta-examples}
Let $f:\zo^n \rightarrow [0,1]$ be a submodular function. There exists an algorithm, that given any $\eps > 0$ and access to random and uniform examples of $f$, with probability at least $5/6$, finds a set of variables $I$ of size $\tilde{O}(1/\eps^{5})$ such that there exists
a submodular $J$-junta $h$ for $J \subseteq I$ of size $\tilde{O}(1/\eps^{2})$ satisfying $\|f -h\|_1 \leq \eps$. The algorithm runs in time $\tilde{O}(n^2 /\eps^{10})$ and uses $\tilde{O}(\log (n) / \eps^{10})$ examples.
\end{corollary}

For general low-influence functions we do not expect to be able to find the influential variables efficiently using random examples alone. For example, Boolean $k$-juntas have total $\ell_1$-influence of at most $k$ but finding the influential variables in $n^{o(k)}$ time is a notoriously hard open problem. However in the special case of monotone functions it is well-known that the influential variables can be found efficiently from random examples alone \cite{Servedio:04mondnf}. The detection of influential variables is based on a simple relationship between $\ell_1$-influences of a monotone (and even unate) function and its degree-1 Fourier coefficients.
\begin{lemma}
Let $f$ be a unate real-valued function. Then for every $i\in [n]$, $$ |\hat{f}(\{i\})| =  \infl_i^{1}(f) .$$
\end{lemma}
\begin{proof}
By definition, $\infl_i^{1}(f) = \E[|\partial_i f|]/2$. For a unate $f$, $\partial_i f$ is either non-negative for all $x$ or non-positive for all $x$. Therefore
$$\infl_i^{1}(f) = \E[|\partial_i f|]/2 = |\E[\partial_i f]|/2 = |\hat{f}(\{i\})| .$$
\end{proof}

Therefore to find influential variables it is sufficient to estimate the degree-1 Fourier coefficients (in the same way as in the proof of Thm.~\ref{thm:find-junta-examples}). As an immediate corollary of this observation and Cor.~\ref{cor:lowsens-junta-l1-spectral-bound} we get the following algorithm.
\begin{corollary}
\label{cor:find-junta-lowsens-examples}
Let $f:\zo^n \rightarrow [0,1]$ be any function. There exists an algorithm, that given any $\eps > 0$ and access to random and uniform examples of $f$, with probability at least $5/6$, finds a set of variables $I$ of size $2^{O(\infl^1(f)/\eps^{2})}$ such that there exists a function $p$ of Fourier degree $2\cdot \infl^1(f)/\eps^{2}$ over variables in $I$ satisfying $\|f - p\|_2 \leq \epsilon$. The algorithm runs in time $\tilde{O}(n) \cdot 2^{O(\infl^1(f)/\eps^{2})}$ and uses $\log (n) \cdot 2^{O(\infl^1(f)/\eps^{2})}$ examples.
\end{corollary}

\subsection{Proper PAC Learning of Submodular Functions} 
In this section we use our junta approximation result and the algorithm for finding the influential variables to get a proper learning algorithm for submodular functions. The previous result on approximation by juntas \cite{FeldmanKV:13} only gives a doubly exponential $2^{2^{O(1/\eps^2)}}$ dependence of running time on $\eps$. This algorithm also serves as a step in our PMAC learning algorithm.

\begin{theorem}
\label{thm:proper-learning}
There exists an algorithm $\A$ that given $\eps > 0$ and access to random uniform examples of any submodular $f:\zon\rightarrow [0,1]$, with probability at least $2/3$, outputs a {\bf submodular} function $h$, such that $\|f-h\|_1 \leq \epsilon$. Further, $h$ is a $J$-junta for some $J$ of size $O(1/\epsilon^2 \cdot \log (1/\epsilon))$ variables, $\A$ also returns $J$ and runs in time $\tilde{O}(n^2/\eps^{10}) + 2^{\tilde{O}(1/\eps^2)}$ and uses $\tilde{O}(\log(n)/\eps^{10}) + 2^{\tilde{O}(1/\eps^2)}$ random examples.
\end{theorem}
\paragraph{The proper learning algorithm.}
\begin{enumerate}
\item Run the algorithm from Cor.~\ref{cor:find-junta-examples} to find a set of variables $I$ of size $s$ such that there exists a submodular $t$-junta $g$ over variables in $I$ satisfying $\|f -g\|_1 \leq \eps/2$ (with probability at least $5/6$).
\item Request $m$ random examples: $(x^1,f(x^1)),(x^2,f(x^2)),\ldots,(x^m,f(x^m))$.
\item FOR every subset $J\subseteq I$ of size $t$ DO
\begin{enumerate}
\item Solve an LP to find a $J$-junta $h:\zo^n \rightarrow [0,1]$ that minimizes $\frac{1}{m}\sum_{i\leq m} |f(x^i) - h(x^i)|$ with constraints requiring that $h$ be submodular.
\item If $\frac{1}{m}\sum_{i\leq m} |f(x^i) - h(x^i)| \leq 3\eps/4$ then return $h$, $J$ and terminate.
\end{enumerate}
\item Return $h \equiv 0$.
\end{enumerate}
\begin{proof}
The specific choices of $s =\tilde{O}(1/\eps^{5})$ and $t = O(1/\epsilon^2 \cdot \log (1/\epsilon))$ are determined by Cor.~\ref{cor:find-junta-examples}. We choose the number of examples $m$ so as to ensure that, with probability at least $5/6$, for every $J$-junta $h$ such that $J \subseteq I$ and $|J| = t$, \equ{\left|\E[|f(x)-h(x)|] - \frac{1}{m}\sum_{i\leq m} |f(x^i) - h(x^i)| \right| \leq \frac{\eps}{4}\ .\label{eq:uniform-conv}}
Standard uniform convergence bounds \cite{Vapnik:98}
\eat{
\footnote{Alternatively one can obtain the bound as follows: (1) it is sufficient to obtain equation (\ref{eq:uniform-conv}) with $\eps/8$ in place of $\eps/4$ for all functions $h$ whose values are discretized with granularity $\eps/8$; (2) there are at most $(8/\eps)^{2^t}$ different $J$-juntas with such granularity and range $[0,1]$; (3) we can use the union bound over Chernoff bounds applied to each individual function to obtain essentially the same bound on $m$.}}
imply that for any fixed set $J$, using $O(2^t/\eps^2 \cdot \log(1/\delta))$ examples will suffice to make sure that equation (\ref{eq:uniform-conv}) holds with probability at least $1-\delta$ for all $J$-juntas with range $[0,1]$. Using the union bound over all ${s \choose t} = 2^{\tilde{O}(1/\eps^2)}$ subsets of $I$ we get that $m=2^{\tilde{O}(1/\eps^2)}$ will suffice to achieve the desired guarantee.

Now, by Cor.~\ref{cor:find-junta-examples}, there exist $J' \subseteq I$ of size $t$ and a submodular $J'$-junta $g$ that satisfies $\|f -g\|_1 \leq \eps/2$. By equation (\ref{eq:uniform-conv}), $\frac{1}{m}\sum_{i\leq m} |f(x^i) - g(x^i)| \leq \|f -g\|_1+\eps/4 \leq 3\eps/4$. This implies that when $J=J'$ the solution of LP will be returned as a hypothesis and the algorithm will not reach Step (4) (assuming that Step (1) is successful and equation (\ref{eq:uniform-conv}) holds).

For any $h$ returned as a hypothesis, $\frac{1}{m}\sum_{i\leq m} |f(x^i) - h(x^i)| \leq 3\eps/4$ and therefore by equation (\ref{eq:uniform-conv}), $\|f -h\|_1 \leq \frac{1}{m}\sum_{i\leq m} |f(x^i) - h(x^i)| +\eps/4 \leq \eps$. This implies that if Step (1) is successful and equation (\ref{eq:uniform-conv}) holds then the algorithm will output a hypothesis $h$ with $\ell_1$-error of at most $\eps$. These conditions hold with probability at least $5/6$.

For a $t$-junta $h$, the minimization of $\ell_1$-error on examples, submodularity and range $[0,1]$ can all be expressed in a linear program with $O(t^2 \cdot 2^t)$ constraints on the values of $h$ at $2^t$ points. The solution to this LP can be found in time $2^{O(t)}$. Therefore the total running time of Step (3) is ${s \choose t} \cdot 2^{O(t)} =  2^{\tilde{O}(1/\eps^2)}$. Combining this with the bounds on the running time and the number of examples from Cor.~\ref{cor:find-junta-examples} we get the claimed bounds.
\end{proof}

\subsection{PMAC Learning of Submodular Functions}
\label{sec:PMAC-learning}
We now show that approximation by a junta can also be used to obtain a PMAC learning algorithm for submodular functions. Our algorithm is based on a reduction from multiplicative approximation to additive approximation. Specifically we use the additive approximation algorithm (Thm.~\ref{thm:proper-learning}) to find a function $g:\zon \rightarrow [0,1]$ over a set of variables $J \subseteq [n]$ that has low $\ell_1$-error. We then prove that,  for at least $1/10$ fraction of values $z \in \zo^J$, $g$ gives a multiplicative approximation to $f$ on at least $1-\eps$ fraction of points $(z,y)$ for $y \in \zo^{\bar{J}}$. This reduces the problem to finding multiplicative approximation to $f$ for values $z$ where the above guarantee does not hold. In other words, we reduce the problem to $\frac{9}{10} 2^{|J|}$ instances of the same problem on a subcube of $\zon$ and execute our algorithm recursively for each of those instances. Importantly, this step solves the problem on $1/10$-fraction of all the inputs and therefore the depth of the recursion needs to be at most $O(\log{(1/\eps)})$. This makes the total number of executions of this procedure upper bounded by $2^{O(|J| \cdot \log{(1/\eps)})}$.

A crucial property of submodular functions that is needed for this reduction step to work is that when a (non-negative) submodular function $f$ is scaled so that $\|f\|_\infty = 1$, then $f$ equals at least some constant $c_1$ on at least a constant fraction of inputs. We obtain this property from the following lemma (from \cite{FeigeMV:07}).
\begin{lemma}
\label{lem:submod-norms}
Let $f:\zon \rightarrow \RR_+$ be a submodular function. Then $\|f\|_1 \geq \frac14 \|f\|_\infty$.
\end{lemma}
We note that this property also holds for XOS functions (with $\frac12$ instead of $\frac14$; see \cite{Feige:06}).
Together with Chernoff-Hoeffding's bound, Lemma \ref{lem:submod-norms} also implies the following lemma.
\begin{lemma}

\label{lem:submod-sample}
There is a constant $c>0$ such that for any submodular $f:\zon \rightarrow \RR_+$, $\gamma \in (0,1)$, any integer $t$, and $t$ points $x^1,x^2,\ldots,x^t$ drawn randomly and uniformly from $\zon$, it holds that $$\Pr\left[\left|\fr{t}\sum_{i\in [t]} f(x^i) - \E[f]\right| \geq \gamma \E[f] \right] \leq 2 e^{-c t \gamma^2} .$$
\end{lemma}

\begin{proof}
Suppose $\|f\|_\infty = M$, hence $f(x^i)$ are independent random variables in $[0,M]$. By Chernoff-Hoeffding bounds,
$\Pr[| \frac{1}{t} \sum_{i=1}^{t} f(x^i) - \E[f]| > \delta M] \leq 2 e^{-c' \delta^2 t}$ for some $c'>0$. We also have $\E[f] = \|f\|_1 \geq \frac14 M$, hence we can set $\delta = \frac14 \gamma$ and the lemma follows.
\end{proof}

We now present the details of the algorithm and its analysis.
\begin{theorem}[Thm.~\ref{thm:pmac-learn-submod-intro} restated]
\label{thm:pmac-learn-submod}
There exists an algorithm $\A$ that given $\gamma,\eps \in (0,1]$ and access to random and uniform examples of any submodular function $f:\zo^n \rightarrow \RR_+$, with probability at least $2/3$, outputs a function $h$ which multiplicatively $(1+\gamma,\eps)$-approximates $f$ (over the uniform distribution). Further, $\A$ runs in time $\tilde{O}(n^2) \cdot 2^{\tilde{O}(1/(\eps\gamma)^2)}$ and uses $\log(n) \cdot 2^{\tilde{O}(1/(\eps\gamma)^2)} $ examples.
\end{theorem}
\eat{
\paragraph{The PMAC learning algorithm.}
}
\begin{proof}
The algorithm $\A$ relies on the reduction we outlined above. Let $\A'(k)$ denote the execution of the learning algorithm at the $k$-th level of recursion. $\A$ executes $\A'(0)$ and $\A'(k)$ is the following algorithm.
\begin{enumerate}
\item If $k\geq 10\log(1/\eps)$ then $\A'(k)$ {\bf returns} the hypothesis $h \equiv 0$.
\item Otherwise, the algorithm estimates $\E[f]$ to within a multiplicative factor of $6/5$ (with probability at least $1-\delta$ for $\delta$ to be defined later). Let $\mu$ denote the obtained estimate (that is, $\E[f] \leq \mu \leq \frac{6}{5} \E[f]$). If $\mu = 0$ the algorithm {\bf returns} $h \equiv 0$. Otherwise, we define the function $f' = f/(4\mu)$.  By Lemma \ref{lem:submod-norms} we know that $\|f'\|_\infty = \|f\|_\infty/(4\mu) \leq 4 \E[f]/(4\mu) \leq 1$.
\item We run our $\ell_1$-learning algorithm from Theorem \ref{thm:proper-learning} on random examples of $f'$ with accuracy $\eps' = \gamma\eps/2400$ and confidence $1-\delta$ (using the standard confidence boosting technique). Let $g$ be the hypothesis output by the algorithm and $J$ be the set of indices of variables it depends on. We treat $g$ as a function on $\zo^J$.
\item We define the output hypothesis $h$ by defining for every $z \in \zo^J$ a function $h_z : \zo^{\bar{J}} \rightarrow \RR_+$ and then setting $h(x) = h_{x_J}(x_{\bar{J}})$. For $z \in \zo^J$, $h_z$ is defined as follows:
\begin{enumerate}
\item \label{item:cond-1} If $g(z) \geq 1/20$ and $\E_{y\in \zo^{\bar{J}}}[|g(z)-f'(z,y)|] \leq 20\eps'$ then define $h_z \equiv (4\mu)(1+\gamma/60)g(z)$
\item Otherwise, execute $\A'(k+1)$ on function $f_z$ over $\zo^{\bar{J}}$ defined as $f_z(y) = f(z,y)$. Let $h_z$ be the output of this execution.
\end{enumerate}
To simulate random examples of $f_z$ we draw random examples of $f$ until an example $(x,\ell)$ is obtained such that $x_J = z$. Note that we cannot find $\E_{y\in \zo^{\bar{J}}}[|g(z)-f'(z,y)|]$ exactly but an estimate within $\eps'/2$ with probability $1-\delta$ would suffice (with minor adjustments in the constants).

$\A'(k)$ {\bf returns} the hypothesis $h$.
\end{enumerate}

We now prove the correctness of the algorithm under the assumption that random estimations and executions of the algorithm from Theorem \ref{thm:proper-learning} are successful. First we observe that if the condition in step (\ref{item:cond-1}) holds then $h_z$ multiplicatively $(1+\gamma,\eps/2)$-approximates $f_z$ over the uniform distribution on $\zo^{\bar{J}}$. By Markov's inequality, the condition  $\E_{y\in \zo^{\bar{J}}}[|g(z)-f'(z,y)|] \leq 20\eps' = \eps\gamma/120$ implies that $$\Pr_{y\in \zo^{\bar{J}}}[|g(z)-f'(z,y)| \geq \gamma/60] \leq \eps/2\ .$$ This means that on all but $\eps/2$ fraction of points $y$, $|g(z)-f'(z,y)| \leq \gamma/60$. On those points $g(z) + \gamma/60 \geq f'(z,y)$. In addition, $f'(z,y) \geq g(z) - \gamma/60 \geq 1/20 - \gamma/60 \geq 1/30$ and therefore $g(z)+ \gamma/60 \leq f'(z,y) + \gamma/30 \leq (1+\gamma) f'(z,y)$. This implies that $g(z)+ \gamma/60$ multiplicatively $(1+\gamma,\eps/2)$-approximates $f'(z,y)$.
By our definition $h_z \equiv (4\mu)(1+\gamma/60)g(z)$ and $f_z(y) = (4\mu) f'(z,y)$.

Now we observe that we can partition the domain $\zon$ into two sets of points:
 \begin{enumerate}
 \item Set $G$ where either $\mu = \E[f] = 0$ or the value output by the hypothesis is $(4\mu)(1+\gamma/60)g(z)$, where $g$ is returned by one of the invocations of the additive approximation algorithm;
 \item The set of points where the recursion reached depth $k>10\log(1/\eps)$.
 \end{enumerate}
By the construction, the points in $G$ can be divided into disjoint subcubes such that in each of them the conditional probability that the hypothesis we output does not satisfy the multiplicative guarantee is at most $\eps/2$. Therefore the hypothesis does not satisfy the multiplicative guarantee on at most fraction $\eps/2$ of the points in $G$.

To finish the proof of correctness it suffices to show that $\Pr[x \not\in G] \leq \eps/2$. To establish this we prove that the fraction of points on which $\A'(k)$ is invoked is at most $(9/10)^{-k}$. We prove this by induction (with $k=0$ being obvious). For any $k < 10\log(1/\eps)$ if $\mu = 0$ then $\A'(k+1)$ is not invoked. Otherwise, we know that $g$ satisfies
$$\E[|f'(x) - g(x_J)|] = \E_{z \sim \zo^J} \left[\E_{y\in \zo^{\bar{J}}}[|f'(z,y) - g(z)|]\right] \leq \eps'\ .$$
Therefore by Markov's inequality,
\equ{\Pr_{z \sim \zo^J} \left[\E_{y\in \zo^{\bar{J}}}[|f'(z,y) - g(z)|] \geq 20\eps' \right]  \leq 1/20\label{eq:high-error-l1} .}
We also know that $$\E_{z \sim \zo^J}[g(z)] \geq \E[f'(x)] - \eps' \geq \frac{\E[f(x)]}{4\mu} -\eps' \geq \frac{5}{24} - \frac{\eps\gamma}{2400} > \frac{1}{5}\ .$$ At the same time $g$ has range $[0,1]$ and hence
$$\E_{z \sim \zo^J}[g(z)] \leq \fr{20}\Pr_{z \sim \zo^J}[g(z) < 1/20] + \Pr_{z \sim \zo^J}[g(z) \geq 1/20]\ .$$
This implies that $\Pr_{x \sim \zo^J}[g(z) \geq 1/20] \geq 1/5 - 1/20 = 3/20$. Together with equation (\ref{eq:high-error-l1}) this implies that the fraction of $z$'s for which both $g(z) \geq 1/20$ and $\E_{y\in \zo^{\bar{J}}}[|f'(z,y) - g(z)|] \leq 20\eps'$ hold is at least $3/20-1/20 = 1/10$. This implies that $\A'(k+1)$ will be invoked on at most $9/10$-fraction of inputs on which  $\A'(k)$ was invoked, proving the inductive claim.

We now also establish the bounds on the running time and sample complexity of this algorithm. Let $t$ be the bound on the size of $J$ in any of the executions of the additive approximation algorithm (Thm.~\ref{thm:proper-learning}). Note that the size of junta does not depend on $n$ or the confidence parameter $\delta$ and therefore is $\tilde{O}(1/(\eps\gamma)^2)$. Let $r$ be the total number of times $\A'$ is executed. From our correctness analysis we can conclude that $r \leq {(2^t)}^{10\log(1/\eps)} = 2^{\tilde{O}(1/(\eps\gamma)^2)}$. It is sufficient to set $\delta = 1/(9r)$ to ensure that the total probability of failure is at most $1/3$. By Lemma \ref{lem:submod-sample} and Thm.~\ref{thm:proper-learning} it is easy to see that each execution of $\A'(k)$ runs in time $\tilde{O}(n^2) \cdot 2^{\tilde{O}(1/(\eps\gamma)^2)}$ and uses $\log(n) \cdot 2^{\tilde{O}(1/(\eps\gamma)^2)}$ examples (of $f$ restricted to a subcube). Simulating a random example for the execution of $\A'(k)$ requires filtering examples which have $t \cdot k$ variables set to a specific value. This means that simulating $\A'(k)$ requires $2^{jk} = 2^{\tilde{O}(1/(\eps\gamma)^2)}$ times more examples of $f$ than the number of examples required by $\A'(k)$. Altogether all $r$ executions run in $\tilde{O}(n^2) \cdot 2^{\tilde{O}(1/(\eps\gamma)^2)}$ and use $\log(n) \cdot 2^{\tilde{O}(1/(\eps\gamma)^2)}$ examples. Note that we made the standard assumption that manipulating values of $f$ takes $O(1)$ time.
\end{proof}

\subsection{Learning of Low-Sensitivity Functions}
We now show that using variants of well-known techniques we can obtain close-to-optimal PAC
learning algorithms for real-valued functions of low total $\ell_1$-influence. We start by proving a generalization of Theorem \ref{thm:XOS-learning-intro}.
\begin{theorem}[subsumes Thm.~\ref{thm:XOS-learning-intro}]
\label{thm:learn-lowinfluence}
Let $\C^+_a$ be the set of all unate functions with range in $[0,1]$ and total $\ell_1$-influence of at most $a$.
There exists an algorithm $\A$ that given $\eps > 0$ and access to random uniform examples of any  $f\in \C^+_a$, with probability at least $2/3$, outputs a function $h$, such that $\|f-h\|_2 \leq \epsilon$. Further, $\A$ runs in time $\tilde{O}(n) \cdot 2^{O(a^2/\eps^{4})}$ and uses $\log n \cdot 2^{O(a^2/\eps^{4})} $ examples.
\end{theorem}
\begin{proof}
Using Corollary~\ref{cor:find-junta-lowsens-examples} we can find a set of variables $I$  of size $|I| = 2^{O(a/\eps^2)}$ such that there exists a function $p$ of Fourier degree $d = 2a/\eps^{2}$ over variables in $I$ satisfying $\|f - p \|_2 \leq \eps/2$. This function is a linear combination of $m=2^{O(a^2/\eps^{4})}$ parities. Standard uniform convergence bounds \cite{Vapnik:98,BartlettMendelson:02} imply that by using $t=O(m/\eps^2)$ random samples and then solving the least squares  regression over all $m$ parities, (with probability $\geq 5/6$) we will get a function $h$ such that $\|f-h\|_2 \leq \|f-p\|_2 + \eps/2 \leq \eps$.
This step requires $n \cdot \poly(m/\eps)  = n \cdot 2^{O(a^2/\eps^{4})}$ time.
\end{proof}

XOS functions have total $\ell_1$-influence of at most 1 and are monotone. Therefore as an immediate corollary of Thm.~\ref{thm:learn-lowinfluence} we obtain Theorem \ref{thm:XOS-learning-intro}.
\eat{
\begin{corollary}
\label{cor:PAC-XOS}
Let $\C_{XOS}$ be the set of all XOS functions on $\zo^n$ with range $[0,1]$. There exists an algorithm $\A$ that given $\eps > 0$ and access to random uniform examples of any  $f\in \C_{XOS}$, with probability at least $2/3$, outputs a function $h$, such that $\|f-h\|_2 \leq \epsilon$. Further, $\A$ runs in time $\tilde{O}(n) \cdot 2^{O(1/\eps^{4})}$ and uses $\log n \cdot 2^{O(1/\eps^{4})}$ examples.
\end{corollary}
}

It is easy to see that the algorithm for learning XOS functions returns a function that is a $2^{O(1/\eps^2)}$-junta. In addition, as we have noted, Lemma \ref{lem:submod-norms} also holds for XOS functions. Therefore using essentially the same reduction that we used in Theorem \ref{thm:pmac-learn-submod} we can obtain a PMAC learning algorithm for XOS functions with the following guarantees.
\begin{corollary}
\label{cor:pmac-learn-xos}
There exists an algorithm $\A$ that given $\gamma,\eps \in (0,1]$ and access to random and uniform examples of any XOS function $f:\zo^n \rightarrow \RR_+$, with probability at least $2/3$, outputs a function $h$ which multiplicatively $(1+\gamma,\eps)$-approximates $f$ (over the uniform distribution). Further, $\A$ runs in time $\tilde{O}(n) \cdot 2^{2^{\tilde{O}(1/(\eps\gamma)^2)}}$ and uses $\log(n) \cdot 2^{2^{\tilde{O}(1/(\eps\gamma)^2)}} $ examples.
\end{corollary}

\section{Applications to Testing and Agnostic Learning}
\label{sec:additinal-app}
\subsection{Testing}
\label{sec:applications-testing}

As is well-known, proper learning with some additional properties (the learner always returns a submodular hypothesis even if the target is not, or we can verify efficiently whether the hypothesis is submodular) also implies testing from random examples with essentially the same bounds on the running time and the number of examples (in the context of Boolean functions this was first observed in \cite{GoldreichGR98}). Considering our proper learning algorithm (Theorem~\ref{thm:proper-learning}), we obtain the following result.

\begin{corollary}
\label{cor:testing}
There is a testing algorithm that given $\eps>0$ and access to random examples of a function $f:\zo^n \rightarrow [0,1]$,
runs in time $\tilde{O}(n^2/\eps^{10}) + 2^{\tilde{O}(1/\eps^2)}$ and uses $\tilde{O}(\log(n)/\eps^{10}) + 2^{\tilde{O}(1/\eps^2)}$ random examples and distinguishes with probability at least $2/3$ the following two cases:
\begin{enumerate}
\item $f$ is submodular;
\item for any submodular function $h$, $\|f-h\|_1 > \epsilon$.
\end{enumerate}
\end{corollary}

Next, we consider the value query oracle model, where we can query $f(x)$ for any $x \in \on^n$.
Here, we can improve the number of queries from $2^{\tilde{O}(1/\epsilon^2)}$ to $2^{\tilde{O}(1/\epsilon)}$, by plugging in the submodularity tester from \cite{SV11}.
The testing algorithm in \cite{SV11} provides the following guarantee: For $f:\on^n \rightarrow \RR$, it queries $1/\epsilon^{O(\sqrt{n} \log n)}$ values of $f$ and
\begin{itemize}
\item If $f$ is submodular, it returns YES.
\item If $f$ is $\epsilon$-far {\em in Hamming distance} (i.e., for any submodular $h$, $\Pr[f(x) \neq h(x)] > \epsilon$), then it returns NO.
\end{itemize}

We observe that this is a stronger notion of testing than testing in $\ell_1$-distance: If $f$ is $\epsilon$-far from submodular in $\ell_1$-distance, it is also at least $\epsilon$-far in Hamming distance. Hence we can use this tester as a building block in our algorithm.

To obtain a tester in $\ell_1$-distance, we can apply the techniques from above to reduce dimension to $O(\frac{1}{\epsilon^2} \log \frac{1}{\epsilon})$ and apply the testing algorithm of Seshadhri and Vondr\'ak \cite{SV11} in this setting. The testing algorithm works as follows.

\paragraph{The testing algorithm.}
\begin{enumerate}
\item Run the algorithm from Cor.~\ref{cor:find-junta-examples} to find a set of variables $I$ of size $|I| = \poly(1/\epsilon)$ such that if $f$ is submodular then the projected function $g = f_I$ is also submodular and satisfies $\|f-g\|_1 \leq \eps/4$ (with high probability).
\item Run Algorithm~\ref{alg:junta} on function $g$ to further reduce the set of variables to $J \subset I$, $|J| = O(\frac{1}{\epsilon^2} \log \frac{1}{\epsilon})$, such that if $f$ is submodular then the function $h = g_J = f_J$ satisfies $\|f-h\|_1 \leq \eps/2$ w.h.p.
\item Let $y^{(1)},\ldots,y^{(m)}$ be uniformly random samples in $\{0,1\}^{\bar{J}}$ and
define $\tilde{h}(x) = \frac{1}{m} \sum_{i=1}^{m} f(x_J,y^{(i)})$ (an approximation to $h$) where $m = \poly(1/\eps)$. Clearly we can simulate a value query to $\tilde{h}$ by $m$ value queries to $f$.
\item Estimate the distance $\|f-\tilde{h}\|_1$ by taking $\poly(1/\epsilon)$ samples; if $\|f-\tilde{h}\|_1 > 3\eps/4$ then answer NO.
\item Run the submodularity tester of \cite{SV11} on function $\tilde{h}$ with parameter $\epsilon/4$ and return its answer.
\end{enumerate}

\begin{theorem}
\label{thm:testing}
The testing algorithm above runs in time $\poly(n,1/\epsilon) + 2^{\tilde{O}(1/\eps)}$, uses $\poly(1/\epsilon) \log n + 2^{\tilde{O}(1/\epsilon)}$ queries to $f$, and distinguishes with probability at least $2/3$ the following two cases:
\begin{enumerate}
\item $f$ is submodular;
\item for any submodular function $h$, $\|f-h\|_1 > \epsilon$.
\end{enumerate}
\end{theorem}

\begin{proof}
Provided that the input function $f$ is submodular, the algorithm from Corollary~\ref{cor:find-junta-examples} and Algorithm~\ref{alg:junta} successfully finds a subset of variables $J$, $|J| = O(\frac{1}{\epsilon^2} \log \frac{1}{\epsilon})$ such that $h = f_J$ is $\epsilon/2$-close to $f$ (in $\ell_1$). We define the function $\tilde{h}$ by averaging $\poly(1/\epsilon)$ functions $f(x_J, y^{(i)})$ where $y^{(i)}$ are random values in $\{0,1\}^{\bar{J}}$. By Chernoff bounds, $\tilde{h}$ is within a $\poly(\epsilon)$ pointwise error of the true projection $h = f_J$. Since $\|h-f\|_1 \leq \epsilon/2$, we will obtain with high probability $\| \tilde{h} - f \|_1 \leq 3 \epsilon/4$ and the tester will pass Step 4. Moreover, note that $\tilde{h}$ is a submodular function (with probability $1$), since each of the functions $f(x_J, y^{(i)})$ is submodular. Finally, in Step 5, we use the submodularity tester from \cite{SV11} which will confirm that $\tilde{h}$ is a submodular function and answer YES.

Conversely, if the input function $f$ is $\epsilon$-far from submodular, it cannot be the case that the projected function $h = f_J$ is simultaneously $\epsilon/2$-close to submodular and also $\epsilon/2$-close to $f$. Therefore, either the test in Step 4 fails because $\tilde{h}$ is far from $f$, or the tester in Step 5 fails because $\tilde{h}$ is far from submodularity. Either way, we answer NO with high probability.

The time and query complexity of the tester can be analyzed as follows. Step 1 runs in time $\tilde{O}(n^2 / \epsilon^{10})$, and uses $\tilde{O}(\log (n) / \epsilon^{10})$ queries (in this case random examples).  Step 2 runs in time $\poly(1/\epsilon)$ and uses $\poly(1/\epsilon)$ queries, since it is essentially a greedy algorithm on $\poly(1/\epsilon)$ variables. We need to estimate the values of the multilinear extension $F(x)$ within a $\poly(\epsilon)$ additive error, which can be done with $\poly(1/\epsilon)$ queries. In Step 3, we generate $\poly(1/\epsilon)$ random samples in $\{0,1\}^{\bar{J}}$, which takes $n \cdot \poly(1/\epsilon)$ time. In Step 4, we estimate the distance $\|\tilde{h} - f\|_1$ using $\poly(1/\epsilon)$ queries to $f$ and $\tilde{h}$; each query to $\tilde{h}$ can be simulated by $\poly(1/\epsilon)$ queries to $f$. Finally, in Step 5 we run the submodularity tester of \cite{SV11} in dimension $n' = O(\frac{1}{\epsilon^2} \log \frac{1}{\epsilon})$. The running time and query complexity of this tester is $1/\epsilon^{O(\sqrt{n'} \log n')} = 2^{\tilde{O}(1/\epsilon)}$.
Again, we are simulating each query to $\tilde{h}$ by $\poly(1/\epsilon)$ queries to $f$, which is absorbed in the $2^{\tilde{O}(1/\epsilon)}$ query complexity.
\end{proof}

\label{sec:agnostic-learning}
\subsection{Agnostic Learning}
We now give some applications to agnostic learning \cite{KearnsSS:94,Haussler:92} with $\ell_1$-error. Our result generalizes those first obtained for submodular functions in \cite{CheraghchiKKL:12} and in addition uses our junta approximation result to reduce the sample complexity from $n^{O(1/\eps^2)}$ to $\log n \cdot 2^{O(1/\eps^{4})}$.

We start with a brief review of the agnostic model. Agnostic learning generalizes the definition of PAC learning to scenarios where one cannot assume that the input labels are consistent with a function from a given class \cite{Haussler:92,KearnsSS:94} (for example as a result of noise in the labels).
\begin{definition}[Agnostic learning with $\ell_1$-error] 
Let $\F$ be a class of real-valued functions on $\zo^n$ and let $D$ be any fixed distribution on $\zo^n$. For any distribution $P$ over $\zo^n \times [0,1]$, let $\mbox{opt}(P,\F)$ be defined as: $$\mbox{opt}(P,\F) =  \inf_{f \in \F} \E_{(x,\ell) \sim P} [ |\ell - f(x) |] .$$ An algorithm $\A$, is said to agnostically learn $\F$ on $D$ if for every $\epsilon> 0$ and any distribution $P$ on $\zo^n \times [0,1]$ such that the marginal of $P$ on $\zo^n$ is $D$, given access to random independent examples drawn from $P$, with probability at least $\frac{2}{3}$, $\A$ outputs a hypothesis $h$ such that $$\E_{(x,\ell) \sim P} [ |h(x)- \ell| ] \leq \mbox{opt}(P, \F) + \epsilon.$$
\end{definition}

We now describe our agnostic learning algorithm for functions with low total $\ell_1$-influence. The algorithm is essentially a Least-Absolute-Error LP (or $\ell_1$-regression) over low-degree monomials with an additional constraint on $\ell_1$-norm of the coefficient vector of the polynomial (which is also its spectral $\ell_1$-norm). The sample complexity analysis of this algorithm is based on the following uniform convergence result of Kakade, Sridharan and Tewari \cite{KakadeST:08} which we include here for completeness.
\begin{theorem}[\cite{KakadeST:08}]
\label{thm:kst}
For $N,W,R > 0$, Let $\B_1^N(W) = \{ w \in \RR^N \cond \|w\|_1 \leq W\}$ and $\B_\infty^N(R) = \{ z \in \RR^N \cond \|z\|_\infty \leq R\}$. Let $P$ be a distribution over $\B_\infty^N(R) \times \RR$. Let $\{(z^1,y^1),\ldots,(z^t,y^t)\}$ be a set of $t$ i.i.d.~samples from $P$. Then, with probability at least $1-\delta$ over the choice of samples, it holds that:
$$\forall w \in \B_1^N(W),\   \left| \E_{(z,y) \sim P}[ |\la z, w \ra - y| ] - \frac{1}{t} \sum_{i=1}^t  |\la z^i, w \ra - y^i| \right| \leq 4 \cdot W R \cdot \sqrt{\frac{\log(N/\delta)}{t}} .$$
\end{theorem}
\begin{theorem}
Let $\C_a$ be the class of all functions with range in $[0,1]$ and total $\ell_1$-influence of at most $a$. There exists an algorithm that learns $\C_a$ agnostically with $\ell_1$-error, runs in time $n^{O(a/\eps^2)}$ and uses $\log (n) \cdot 2^{O(a^2/\eps^{4})} $ examples.
\end{theorem}
\begin{proof}
Let $P$ be any distribution over $\zon \times [0,1]$ whose marginal distribution over $\zon$ is uniform and let $g^*\in \C_a$ be such that $\E_{(x,\ell) \sim P}[|g^*(x)-\ell|] = \min_{g\in \C_a}\{\E_{(x,\ell) \sim P}[|g(x)-\ell|]\} = \Delta$. By Corollary \ref{cor:lowsens-junta-l1-spectral-bound}, we know that there exists a function $p$ of Fourier degree $d = O(a/\epsilon^2)$ such that $\|p-g^*\|_1 \leq \eps/2$ and $W = \|\hat{p}\|_1 = 2^{O(d^2)}$.

We draw $t$ examples $\{(x^i,\ell^i)\}_{i\leq t}$ where $t$ is chosen so as to ensure that, with probability $\geq 5/6$, the maximum of $$\left|\E_{P}[|p'(x)-\ell|] - \fr{t} \sum_{i\leq t} |p'(x^i) - \ell^i| \right|$$ taken over all functions $p'$ of Fourier degree $d$ and spectral $\ell_1$-norm $\leq W$ is at most $\eps/4$.

Now we formulate an LP over coefficients $\{\alpha_{S}\}_{S\subseteq [n],\ |S| \leq d}$ for all parities of degree at most $d$ that minimizes $$\sum_{i\leq t} \left| \sum_{S\subseteq [n],\ |S| \leq d} \alpha_S \chi_S(x^i) - \ell^i\right| $$ subject to
$\sum_{S\subseteq [n],\ |S| \leq d} |\alpha_S| \leq W$. Let $p'(x)$ be the function obtained by solving this LP.
By the definition of our LP, \equ{\sum_{i\leq t} |p'(x^i) - \ell^i| \leq \sum_{i\leq t} |p(x^i) - \ell^i|\ .\label{eq:lp-bound}}
Both $p$ and $p'$ are functions of Fourier degree $d$ and spectral $\ell_1$-norm $\leq W$ and therefore by our choice of $t$, with probability at least $2/3$, $$\E_{P}[|p(x)-\ell|] \geq \fr{t} \sum_{i\leq t} |p(x^i) - \ell^i| - \eps/4\ $$ and $$\E_{P}[|p'(x)-\ell|] \leq \fr{t} \sum_{i\leq t} |p'(x^i) - \ell^i| + \eps/4\ .$$
This implies that $$\E_{P}[|p'(x)-\ell|] -\E_{P}[|p(x)-\ell|] \leq \fr{t} \sum_{i\leq t} |p'(x^i) - \ell^i| - \fr{t} \sum_{i\leq t} |p(x^i) - \ell^i| + \eps/2.$$
By combining this with eq.~(\ref{eq:lp-bound}), we get that $\E_{P}[|p'(x)-\ell|] \leq \E_{P}[|p(x)-\ell|] +\eps/2$ and hence our hypothesis $h(x) = p'(x)$ satisfies
\alequn{\E_P[|p'(x)-\ell|] & \leq
\E_{P}[|p(x)-\ell|] + \eps/2 \\ & \leq
\E_{P}[|g^*(x)-\ell|] + \|g^* - p\|_1 + \eps/2 \leq \E_{P}[|g^*(x)-\ell|]  + \eps = \Delta +\eps .}

Finally, to bound $t$ we use Thm.~\ref{thm:kst}. Note that the dimension $N= |\{ S \cond S\subseteq [n], |S| \leq d \}| \leq n^d$, $\ell_1$ constraint on the sum of coefficients is $W$ and $R = 1$ since the range of each parity function is $\{-1,1\}$. This implies that taking $t = O(W^2 \log m /\eps^2 ) = 2^{O(d)} \log n$ examples will ensure uniform convergence with error $\eps/4$ and confidence $5/6$. The running time of solving the LP is $n^{O(d)}$.
\end{proof}

Given access to value queries one can make agnostic learning more efficient. We first remark that if one is only concerned with squared error, then agnostic learning of all functions with spectral $\ell_1$-norm of $W$ can be done in time $\poly(n,W,1/\eps)$ using the algorithm of Kushilevitz and Mansour \cite{KushilevitzMansour:93} (see \cite{FeldmanKV:13} for details). Achieving agnostic learning for $\ell_1$-error is substantially more involved. This problem was solved by Gopalan, Kalai and Klivans \cite{GopalanKK:08} who proved the following theorem.
\begin{theorem}[Sparse $\ell_1$-regression \cite{GopalanKK:08}]
\label{th:gkk}
For $W>0$, we define $\C_W$ as $\{ p(x) \cond \|\hat{p}\|_1 \leq W\}$. There exists an algorithm $\A$ that given $\eps > 0$ and access to value queries for any real-valued $f:\zo^n\rightarrow [-1,1]$, with probability at least $2/3$, outputs a function $h$, such that $\|f-h\|_1 \leq \Delta + \epsilon$, where $\Delta = \min_{p\in \C_W}\{\|f-p\|_1\}$. Further, $\A$ runs in time $\poly(n,W,1/\eps)$
\end{theorem}
By Corollary \ref{cor:lowsens-junta-l1-spectral-bound}, we know that every function with range in $[0,1]$ and total $\ell_1$-influence of at most $a$ can be $\eps$-approximated in $\ell_1$ norm by a function of spectral $\ell_1$-norm $2^{O(a^2/\eps^4)}$. Together with Theorem \ref{th:gkk} this implies the existence of the following learning algorithm.
\begin{theorem}
\label{th:agn-learn-lowsens-valueq}
Let $\C_a$ be the class of all functions with range in $[0,1]$ and total $\ell_1$-influence of at most $a$. There exists an algorithm that, given access to value query oracle, learns $\C_a$ agnostically with $\ell_1$-error in $\poly(n) \cdot 2^{O(a^2/\eps^4)}$ time.
\end{theorem}
We remark that for the special case of submodular functions, a slightly faster ($\poly(n) \cdot 2^{O(1/\eps^2)}$-time) and attribute-efficient algorithm was given in \cite{FeldmanKV:13}. 

\section{Discussion and Open Problems}
Our paper essentially resolves the question of additive approximation by juntas for submodular functions and multiplicative approximation by juntas for monotone submodular functions. However many natural questions and gaps in our bounds still remain. The most obvious one is whether there exists a multiplicative approximation junta for non-monotone submodular functions similar to the monotone case (of size $\poly(1/\gamma, \log (1/\epsilon))$ as in Theorem~\ref{thm:PMAC-junta}). We note that the existence of a multiplicative $(1+\gamma,\epsilon)$-approximation junta of size $2^{\poly(1/\gamma,1/\epsilon)}$ follows from our PMAC-learning algorithm for non-monotone submodular functions (Section~\ref{sec:PMAC-learning}). It is an interesting question whether there is indeed a significant gap between monotone and non-monotone submodular functions or not.  Another natural question is whether every XOS function can be multiplicatively $(1+\gamma,\epsilon)$-approximated by a junta of exponential in $1/\eps$ and $1/\gamma$ size.

It would also be interesting to understand under what distributional assumptions (besides uniform/product) such strong approximation-by-junta results hold. Algorithmically, it is interesting whether a junta of close to optimal size (that is, $\tilde{O}(1/\eps^2)$) can be found in polynomial time given random examples alone. Our algorithm in Theorem \ref{thm:find-junta-examples} only finds a larger $\tilde O(1/\eps^5)$-junta in polynomial time.


\section*{Acknowledgements}
We would like to thank Seshadhri Comandur and Pravesh Kothari for useful discussion. We also thank the anonymous FOCS and SICOMP referees for their comments and useful suggestions. 

\bibliographystyle{alpha}

\bibliography{submodjuntarefs}

\appendix

\section{Submodular Functions over General Product Distributions}
\label{sec:product-distribution}
In this section, we address the question of extending our results to general product distributions on $\zo^n$.
We present our submodular junta approximation result in this more general setting to illustrate how parameters of the distribution affect the statement of the results. Replicating all our results in this general setting is beyond the scope of this work. We note that the extension of the Fourier analysis based tools that we use to this setting is well-known.
\begin{theorem}
\label{thm:submod-junta-product}
Let $\cD$ be a product distribution on $\zo^n$ and let $p_0 = \min_{i \in [n], a \in \zo} \Pr_{x \sim \cD}[x_i = a] > 0$.
Then for any $\epsilon \in (0,\frac12)$ and any submodular function $f:\zo^n \rightarrow [0,1]$, there exists a submodular function $g:\zo^n \rightarrow [0,1]$ depending only on a subset of variables $J \subseteq [n]$, $|J| = O(\frac{1}{p_0 \epsilon^2} \log \frac{1}{p_0 \epsilon})$, such that $\E_{x \sim \cD}[|f(x) - g(x)|^2] \leq \epsilon^2$.
\end{theorem}

Note the dependence of the size of the junta on $p_0$. Let us show first that a factor of $\Omega(1/p_0)$ is necessary.

\begin{propos}
Let $s \geq 2$ be even and let $\cD$ be a product distribution on $\zo^n$ such that $\Pr_{x \sim \cD}[x_i = 1] = p_0 = 1 - (1/2)^{2/s}$ for all $i \in [n]$. Let $f(x) = \min \{ \sum_{i \in S} x_i, 1 \}$ where $|S| = s$.
Then there is no function $g:\zo^n \rightarrow \RR$ such that $\E_{x \sim \cD}[|f(x)-g(x)|^2] < 1/8$ and $g$ depends on fewer than $s/2$ variables.
\end{propos}

Note that $s = -2 / \log_2 (1-p_0) = \Omega(1/p_0)$, so the claim is that we need $\Omega(1/p_0)$ variables to approximate $f$ even within a constant $\ell_2$ error.
To prove this, consider any function $g$ depending on $|J|=s/2$ variables. Variables outside of $S$ do not affect $f$ so we may assume that $J \subset S$. Note that $f(x)$ attains only values $0$ and $1$. With respect to $\cD$, we have $\Pr_{x \sim \cD}[f(x) = 1] = \Pr_{x \sim \cD}[\exists i \in S; x_i = 1] = 1 - (1-p_0)^s = 1 - 1/4 = 3/4$. Furthermore, $\Pr_{x \sim \cD}[\forall i \in J; x_i = 0] = (1-p_0)^{s/2} = 1/2$. Conditioned on $x_i=0$ for all $i \in J$, we have $f(x) = 0$ with probability $(1-p_0)^{s/2} = 1/2$ and $f(x) = 1$ with probability $1/2$ . However, $g(x)$ has the same value in all these cases and hence we get $\E[|f(x) - g(x)|^2 \mid x_J = {0}] \geq 1/4$ (the best choice is to set $g(x) = 1/2$ whenever $x_J = {0}$). Since $\Pr_{x \sim \cD}[x_J = {0}] = 1/2$, we obtain $\E_{x \sim \cD}[|f(x)-g(x)|^2] \geq 1/8$.

\

Next, we turn to the proof of Theorem~\ref{thm:submod-junta-product}. As in the case of uniform distributions, it is sufficient to prove the following lemma which is then iterated to obtain Theorem~\ref{thm:submod-junta-product}.

\begin{lemma}
\label{lem:submod-product-reduce}
Let $\cD$ be a product distribution on $\zo^n$ such that $p_0 = \min_{i \in [n], a \in \zo} \Pr_{x \sim \cD}[x_i = a] > 0$.
For any $\epsilon \in (0,\frac12)$ and any submodular function $f:\{0,1\}^J \rightarrow [0,1]$,
there exists a submodular function $h:\{0,1\}^J \rightarrow [0,1]$ depending only on a subset of variables $J' \subseteq J$, $|J'| = O(\frac{1}{p_0 \epsilon^2} \log \frac{|J|}{\epsilon^2})$, such that $\E_{x \sim \cD}[|f(x) - h(x)|^2] \leq \frac14 \epsilon^2$.
\end{lemma}

In the following, we prove this lemma. Again, our proof relies on a greedy procedure to select the significant variables, and the boosting lemma to obtain a high-probability bound on the event that the function is sufficiently Lipschitz in the remaining variables. We need the following general version of the boosting lemma (\cite{GoemansVondrak:06}, somewhat reformulated here).

\begin{lemma}[non-uniform boosting lemma]
\label{lem:boosting-asymmetric}
Let $\cF \subseteq \zo^X$ be down-monotone and $\eta \in (0,1)$. Let $\cD, \cD'$ be product distributions on $\zo^X$ where $\Pr_{x \sim \cD'}[x_i = 0] = (\Pr_{x \sim \cD}[x_i = 0])^\eta$ for each $i \in X$. Then
$$ \Pr_{x \sim \cD'}[x \in \cF] \geq \left(\Pr_{x \sim \cD}[x \in \cF] \right)^\eta.$$
\end{lemma}

We choose the significant variables using the following algorithm.

\begin{algorithm}
\label{alg:prod-junta}
Given $f:\{0,1\}^J \rightarrow [0,1]$ and a product distribution $\cD$ on $\zo^J$, produce a small set of important coordinates $J'$ as follows (for parameters $\alpha,\eta>0$):
\begin{itemize}
\item Let $\cD_0$ be a product distribution such that $\Pr_{x \sim \cD_0}[x_i=0] = (\Pr_{x \sim \cD}[x_i=0])^\eta$ for each $i \in J$, and $\cD_1$ a product distribution such that $\Pr_{x \sim \cD_1}[x_i=1] = (\Pr_{x \sim \cD}[x_i=1])^\eta$ for each $i \in J$.
\item Set $S = T = \emptyset$.
\item As long as there is $i \notin S$ such that $\Pr_{x \sim \cD_0}[\partial_i f(x \wedge \b1_S) > \alpha] > 1/2$, include $i$ in $S$. \\
{\em (This step is sufficient for monotone submodular functions.)}
\item As long as there is $i \notin T$ such that $\Pr_{x \sim \cD_1}[\partial_i f(x \vee \b1_{J \setminus T}) < -\alpha] > 1/2$, include $i$ in $T$. \\
{\em (This step deals with non-monotone submodular functions.)}
\item Return $J' = S \cup T$.
\end{itemize}
\end{algorithm}

Note how the algorithm changed from the case of uniform distributions: the criterion for selecting variables is now based on discrete derivatives with respect to a non-uniformly sampled set, according to a distribution derived in a certain way from the target distribution $\cD$. The goal of this criterion is to achieve the following guarantee.

\begin{lemma}
With the same notation as above, for any $i \in J \setminus J'$
$$ \Pr_{x \sim \cD}[\partial_i f(x \wedge \b1_{J'}) > \alpha] \leq (1/2)^{1/\eta} $$
and
$$ \Pr_{x \sim \cD}[\partial_i f(x \vee \b1_{J \setminus J'}) < -\alpha] \leq (1/2)^{1/\eta}. $$
\end{lemma}

\begin{proof}
Directly from Lemma~\ref{lem:boosting-asymmetric}, applied to events on the subcube $\zo^{J'}$: for a variable that was not selected by the algorithm, we have $\Pr_{x \sim \cD_0}[\partial_i f(x \wedge \b1_{J'}) > \alpha] \leq 1/2$. Since this is a down-monotone event, and $\Pr_{x \sim \cD_0}[x_i=0] = (\Pr_{x \sim \cD}[x_i=0])^\eta$,  Lemma~\ref{lem:boosting-asymmetric} implies $\Pr_{x \sim \cD}[\partial_i f(x \wedge \b1_{J'})] \leq (1/2)^{1/\eta}$. Similarly (by flipping the cube to change an up-monotone event into a down-monotone one), we obtain $\Pr_{x \sim \cD}[\partial_i f(x \vee \b1_{J \setminus J'}) < -\alpha] \leq (1/2)^{1/\eta}$.
\end{proof}

The analysis of the size of set $J'$ is identical to the proof of Lemma~\ref{lem:size}. Compared to Lemma~\ref{lem:size}, the only difference is that we keep track of a random subset of the selected variables, sampled according to the distribution $\cD_0$ (or complement of $\cD_1$, in the second part of the proof). Each selected variable appears in this random set with probability at least $1 - (1-p_0)^\eta \geq \eta p_0$ and hence contributes at least $\frac12 p_0 \eta \alpha$ to its expected value. Therefore, we obtain the following.

\begin{lemma}
The number of variables chosen by the procedure above is $|J'| \leq \frac{4}{p_0 \alpha \eta}$.
\end{lemma}

The rest of the analysis proceeds similarly to the uniform distribution case. We choose $\eta = 1 / \log_2 \frac{16|J|}{\epsilon^2}$ and $\alpha = \frac{1}{16} \epsilon^2$. This ensures that when $x$ is sampled according to $\cD$ restricted to $J'$, the probability that $f(x,y)$  is $\alpha$-Lipschitz in the variables $y$ is at least $1 - 2|J| \cdot (1/2)^{1/\eta} \geq 1 - \frac18 \epsilon^2$. For those points $x$ where $f(x,y)$ is $\alpha$-Lipschitz in $y$, we get by Corollary~\ref{cor:Lipschitz} that the variance of $f$ is at most $2\alpha = \frac18 \epsilon^2$. Therefore, as before we conclude that a junta on the variables indexed by $J'$ approximates $f$ within $\ell_2$ error $\frac12 \epsilon$. The size of the junta is $|J'| = O(\frac{1}{\alpha \eta p_0}) = O(\frac{1}{p_0 \epsilon^2} \log \frac{|J|}{\epsilon^2})$, which proves Lemma~\ref{lem:submod-product-reduce}.

\end{document}